\documentclass[11pt,reqno]{article}

\usepackage{amsmath, amssymb, amscd, amsthm, amsfonts}

\usepackage{graphicx}
\usepackage[perpage,symbol*]{footmisc}
\usepackage{float}
\usepackage{hyperref}
\usepackage{pgfplots}

\usepackage{natbib}
\usepackage{authblk}

\usepackage{algorithm,algcompatible,amsmath}
\algnewcommand\INPUT{\item[\textbf{Input:}]}%
\algnewcommand\OUTPUT{\item[\textbf{Output:}]}%

\usepackage{enumitem}

\usepackage[english]{babel}
\setcitestyle{authoryear,open={(},close={)}}

\usepackage{colortbl,dcolumn}

\usepackage{listings} 
\usepackage{xcolor} 

\renewcommand{\raggedright}{\leftskip=0pt \rightskip=0pt plus 0cm}
\raggedright

\def\hilite<#1>{%
	\temporal<#1>{\color{blue!25}}{\color{magenta}}%
	{\color{blue!55}}}

\newcolumntype{H}{>{\columncolor{blue!20}}c!{\vrule}}
\newcolumntype{H}{>{\columncolor{blue!20}}c}

\newtheorem{theorem}{Theorem}

\newtheorem{definition}{Definition}

\newtheorem{lemma}{Lemma}

\numberwithin{equation}{section} 
\numberwithin{theorem}{section}
\numberwithin{lemma}{section} 
\numberwithin{corollary}{section}
\numberwithin{definition}{section}
\numberwithin{proposition}{section} 
\numberwithin{remark}{section}
\numberwithin{example}{section}
\DeclareMathOperator\supp{supp}

\usepackage{fancyhdr}

\renewcommand{\oddsidemargin}{0mm}

\usepackage{setspace}
\def\R{\mathbb R}

\definecolor{mydarkgreen}{rgb}{0,0.4,0}

\makeatletter
\def\@makefnmark{}
\def\qed{\hfill$\square$\smallskip}
\makeatother
\usepackage{amsfonts,amssymb}
\usepackage{mathrsfs}

\begin{document}

\pagestyle{fancy}
\fancyhf{}
\rhead{\textbf{Meng} and \textbf{Eloyan}}
\lhead{\textit{Journal of the Royal Statistical Society: Series B (Statistical Methodology)}}
\fancyfoot[CE,CO]{\thepage}

\vspace{-15in}

	\title{\Large {\textbf{Principal Manifold Estimation via Model Complexity Selection}}}
	\author[1,*]{Kun Meng}
	\author[1]{Ani Eloyan}
		
	\affil[1]{\small Department of Biostatistics, Brown University School of Public Health, Providence, RI 02903, USA}
	\affil[*]{Corresponding Author: e-mail: \texttt{kun\_meng@brown.edu}, Address: 121 S. Main St, Providence, RI 02903, USA.}
	
	\maketitle

\begin{abstract}
	\noindent We propose a framework of principal manifolds to model high-dimensional data. This framework is based on Sobolev spaces and designed to model data of any intrinsic dimension. It includes principal component analysis and principal curve algorithm as special cases. We propose a novel method for model complexity selection to avoid overfitting, eliminate the effects of outliers, and improve the computation speed. Additionally, we propose a method for identifying the interiors of circle-like curves and cylinder/ball-like surfaces. The proposed approach is compared to existing methods by simulations and applied to estimate tumor surfaces and interiors in a lung cancer study. 

\end{abstract}

\noindent%
\footnote{
\begin{itemize}
    \item The project was supported by Grant Number 5P20GM103645 from the National Institute of General Medical Sciences.
    \item This paper is published in the \textit{Journal of the Royal Statistical Society: Series B (Statistical Methodology)}, available on \url{https://doi.org/10.1111/rssb.12416}.
    \item \textbf{Correspondence:} Kun Meng, 121 S. Main St, Providence, RI, 02903, USA; Email: \texttt{kun\_meng@brown.edu}
    \item \textbf{Abbreviations:} CT, computed tomography; HS, Hastie and Stuetzle (1989); iid, identically and independently distributed; KDE, kernel density estimate; KLD, Kullback-Leibler divergence; MSD, mean squared distance; PME, principal manifold estimate; PS, principal surface estimate by Yue et al. (2016); PCA, principal component analysis; PDF, probability density function; sd, standard deviation.
\end{itemize}

}


\section{Introduction}\label{introduction}

\textit{Manifold learning} is a method for modeling high-dimensional data, assuming that data are from a low-dimensional manifold and corrupted by high-dimensional noise. The dimension of the low-dimensional manifold is called the \textit{intrinsic dimension} of data. There are two primary components of manifold learning: (i) \textit{parameterization} - uncovering a low-dimensional description of high-dimensional data; (ii) \textit{embedding} - finding a map relating the low-dimensional description and high-dimensional data. The two components are entangled with each other. Based on a given parameterization, embedding becomes a statistical fitting problem. In turn, projecting data to the image of an embedding map results in a parameterization (e.g., \cite{yue2016parameterization}). In this paper, we propose a framework and estimation approach combining these two components. Specifically, our proposed approach constructs an embedding map from a ``partial" parameterization and obtains a full parameterization from this embedding map. We define principal manifolds as minima of a functional equipped with a regularity penalty term derived as a semi-norm on a Sobolev space.  A Sobolev embedding theorem implies the differentiability of our proposed manifolds. The novel framework of principal manifolds allows the intrinsic dimension of data to be any positive integer. The linear principal component analysis (PCA, \cite{jolliffe1986principal}) and principal curve algorithm (\cite{hastie1989principal}) are special cases of this framework. We provide topological and functional analysis arguments giving mathematical foundations of our proposed principal manifold framework. To avoid overfitting and preserve the curvatures of underlying manifolds, we propose a model complexity selection method. Additionally, this method drastically reduces the computational cost and eliminates the effects of outliers. Based on this method and the theory of reproducing kernel Hilbert spaces, we propose an algorithm to estimate principal manifolds efficiently. Additionally, motivated by a problem in radiation therapy for lung cancer patients, we propose a method for identifying interiors of circle-like curves and cylinder/ball-like surfaces. 

Throughout this paper, we use the following notations: (i) $d$ and $D$, with $d<D$, denote the dimensions of intrinsic manifolds and the spaces into which these manifolds are embedded, respectively. (ii) $\Vert x\Vert_{\R^q}=(\sum_{k=1}^q x_k^2)^{1/2}$ for all $x\in\R^q$ and all positive integers $q$. (iii) Let $q_1,q_2\in\{d,D\}$, $k\in \{1,2,\cdots,\infty\}$, and $I$ be a subset of $\R^{q_1}$, $C^k(I\rightarrow\R^{q_2})$ denotes the collection of $I\rightarrow\R^{q_2}$ maps whose components have up to $k^{th}$ continuous classical derivatives. For simplicity, $C^k\left(I\right)=C^k (I\rightarrow\R^1)$ and $C=C^0$. (iv)  $\delta_x$ is the point mass at $x$ (see Section 6.9 of \cite{rudin1991functional}). (v) $L^p$ and $\Vert\cdot\Vert_{L^p}$, $p\in[1,\infty]$, denote \textit{Lebesgue spaces} and their norms (see Chapter 2 of \cite{adams2003sobolev}).

A considerable amount of work has been done for parameterization and embedding tasks. ISOMAP (\cite{tenenbaum2000global}), locally linear embedding (\cite{roweis2000nonlinear}), and Laplacian eigenmaps (\cite{belkin2003laplacian}) constructed parameterizations of high-dimensional data. \cite{hastie1989principal} (hereafter HS) proposed a \textit{principal curve} framework and algorithm for the embedding task. HS defined principal curves as follows.
\begin{definition}\label{Def 1}
(Part I) Let $I\subset\R^1$ be a closed and possibly infinite interval. Suppose a map $f: I\rightarrow\R^D$ satisfies the conditions (referred to as \textit{HS conditions} throughout this paper): (i) $f\in C^\infty(I\rightarrow\R^D)$; (ii) $\left\Vert f'(t)\right\Vert_{\R^D}=1$ for all $t\in I$, i.e., $f$ is arc-length parameterized; (iii) $f$ does not self intersect, i.e. $t_1\ne t_2$ implies $f(t_1)\ne f(t_2)$; (iv) $\int_{\left\{t: f(t)\in B\right\}} dt<\infty$ for any finite ball $B$ in $\R^D$. Then $\pi_f: \R^D\rightarrow I$ is defined as follows and called the projection index with respect to $f$.
\begin{equation}\label{eq:projection}
\pi_f(x)=\sup\left\{t\in I:\left\Vert x-f(t)\right\Vert_{\R^D}=\inf_{t'\in I}\left\Vert x-f(t')\right\Vert_{\R^D}\right\}, \ \ \mbox{for all $x\in\R^D$.}
\end{equation}

\noindent (Part II) Suppose $X$ is a continuous random $D$-vector with finite second moments. Principal curves of $X$ are all maps $f:I\rightarrow\R^D$ satisfying HS conditions and the self-consistency defined as 
\begin{align}\label{self consistency}
\mathbb{E}\left(X\vert\pi_f(X)=t\right)=f(t).
\end{align}
\end{definition}
\noindent The projection index $\pi_f$ is well-defined under HS conditions. However, HS conditions are restrictive due to the following reasons: 1) condition (ii) requires principal curves to be arc-length parameterized, while the arc-length parameterization is not generalizable to higher dimensions; 2) condition (iii) rules out many curves in applications, e.g., a handwritten ``8'' in handwriting recognition; 3) condition (iv) is not straightforward to verify. Furthermore, the HS principal curve framework has a model bias (see Section 6 of \cite{hastie1989principal}). 

To remove  the model bias in the HS principal curve framework, \cite{tibshirani1992principal} proposed a new
principal curve framework based on a mixture model and self-consistency (\ref{self consistency}). HS
showed that curves satisfying (\ref{self consistency}) are critical points of the \textit{mean squared distance} (MSD) functional $\mathcal{D}_X(f)=\mathbb{E}\left\Vert X-f\left(\pi_f(X)\right)\right\Vert^2_{\R^D}$. However, \cite{duchamp1996extremal} showed that these critical points may be \textit{saddle} points, i.e., there may exist adjacent curves with smaller MSD than that of curves satisfying (\ref{self consistency}). This \textit{saddle issue}  was a flaw of the frameworks based on (\ref{self consistency}). \cite{gerber2013regularization} explained the saddle issue from the ``orthogonal/along" variation trade-off viewpoint and discussed the challenges stemming from this issue in model complexity selection. To remove the saddle issue, \cite{gerber2013regularization} avoided using MSD $\mathcal{D}_X(f)$ and proposed a new functional $\mathcal{Q}_X(\pi)=\mathbb{E}\left\{\left[\mathbb{E}(X\vert\pi(X))-X\right]^T \frac{d}{dt}\big\vert_{t=\pi(X)}\mathbb{E}(X\vert\pi(X)=t)\right\}$ modeling the parameterization maps $\pi:\R^D\rightarrow I$. This functional penalizes the non-orthogonality between fitting error $\mathbb{E}(X\vert\pi(X))-X$ and curve tangent $\frac{d}{dt}\big\vert_{t=\pi(X)}\mathbb{E}(X\vert\pi(X)=t)$. Principal curves, which satisfy (\ref{self consistency}), correspond to the minima of $\mathcal{Q}
_X(\pi)$. However, the use of $\mathcal{D}_X(f)$ for measuring the discrepancy between data $X$ and fitted $f(\pi_f(X))$ is of interest due to  the interpretability of $\mathcal{D}_X(f)$. Another approach to removing the saddle issue, while using $\mathcal{D}_X(f)$, is to avoid self-consistency and define principal curves by minimizing MSD with a length constraint or a regularity penalty. \cite{kegl2000learning} defined principal curves as the minima $\arg\min_f\{\mathcal{D}_X(f): f\in BV([a,b]), V_a^b(f)\le L \}$, where $L>0$ is pre-defined and $BV([a,b])$ is the collection of functions $f$ on $[a,b]$ with finite total variation $V_a^b(f)<\infty$. However, functions in $BV([a,b])$ are not necessarily everywhere differentiable. Indeed, the algorithm proposed by \cite{kegl2000learning} fits data by polygonal lines, which are only piecewise differentiable. In many applications, we expect globally differentiable curves. Additionally, the \cite{kegl2000learning} framework is only defined for curves, i.e., the intrinsic dimension $d=1$. \cite{smola2001regularized} proposed the \textit{regularized principal manifold} framework, where regularized estimates of principal manifolds are minimizers of the following form.
\begin{align}\label{Smola et al}
\arg\min_{f\in\mathscr{F}}\left\{\mathbb{E}\left\Vert X-f\left(\pi_f(X)\right)\right\Vert_{\mathbb{R}^D}^2+\lambda \Vert Pf\Vert_{\mathscr{H}}^2\right\},
\end{align}
where $\mathscr{F}$ is a collection of functions and $P$ is an operator mapping $f$ into an inner product space $\mathscr{H}$. \cite{smola2001regularized} (Example 7) showed that the \cite{kegl2000learning} definition of principal curves is essentially the special case of (\ref{Smola et al}), where $P=\frac{d}{dt}$ and $\mathscr{H}=L^2$, i.e., the penalty term in (\ref{Smola et al}) is $\Vert f'\Vert_{L^2}^2=\int_a^b\Vert f'\Vert_{\mathbb{R}}^2 dt$ (the derivative $f'$ is defined only almost everywhere with respect to the Lebesgue measure, rather than exactly everywhere). However, the regularized principal manifold approach defined by \cite{smola2001regularized} has several limitations. The problem of selection of the tuning parameter $\lambda$, and more generally, model complexity to avoid overfitting and preserve intrinsic curvatures remains unresolved, as well as a definition of the projecting operator $P$ that would correspond to the tuning parameter selection. HS shows that $\pi_f$ is well-defined under the restrictive HS conditions and when the intrinsic dimension $d=1$, however, there is no discussion of assumptions on the collection $\mathscr{F}$ by \cite{smola2001regularized} such that the resulting $\pi_f$ is well-defined. The function spaces $\mathscr{H}$ and $\mathscr{F}$ compatible with $P$ and $\pi_f$ are not defined. Our proposed principal manifold estimation approach addresses these limitations. 

The paper is organized as follows. In Section \ref{geometry}, we propose a condition for defining the projection indices $\pi_f$ associated with maps $f: \R^d\rightarrow\R^D$, where $d$ is allowed to be any positive integer. This proposed condition is less restrictive and easier to verify than HS conditions. In Section \ref{definitions of principal manifolds}, based on this condition and function spaces of Sobolev type, we define principal manifolds by minimizing MSD equipped with a total squared curvature penalty. This definition solves the differentiability problem in \cite{kegl2000learning}. In Section \ref{Joints of a random vector}, we present the outline of our proposed \textit{principal manifold estimation} algorithm by briefly describing its main steps. Details of these steps are provided in Sections \ref{section: The Reduction Step of PME} and \ref{explicit formula}. Motivated by \cite{eloyan2011smooth}, we propose a data reduction method in Section \ref{section: The Reduction Step of PME}. We then present the PME algorithm in Section \ref{explicit formula}. A simulation study comparing the performance of the PME algorithm to some existing manifold learning methods is presented in Section \ref{simulations}. In Section \ref{Interior classifier}, we propose a method for identifying the interiors of circle-like curves and cylinder/ball-like surfaces. The performance of the proposed method for estimating lung cancer tumor surfaces and interiors using computed tomography (CT) data is presented in Section \ref{An Application}. 


\section{Manifolds and Projection Indices}\label{geometry}

Before defining principal manifolds, we introduce concepts of manifolds and projection indices. 
\begin{definition}\label{Def 2}
	Let $f\in C(\R^d\rightarrow\R^D)$, then the image of $f$, i.e., $M_f^d=\{f(t):t\in\mathbb{R}^d\}$, is called a $d$-dimensional manifold determined by $f$, where $f$ is called an embedding map and $\R^d$ is called the parameter space. Furthermore, $f$ is called a homeomorphism if its inverse $f^{-1}:M^d_f\rightarrow\mathbb{R}^d$ exists and is continuous. Here, the continuity of $f^{-1}$ is associated with the subspace topology of $M_f^d$, i.e., the topology $\{U\bigcap M_f^d:U\mbox{ is an open subset of }\R^D\}$.
\end{definition}
\noindent In applications, $\R^d$ is the space containing latent parameterizations $\{t_i\}_{i=1}^I$ of observed data $\{x_i\}_{i=1}^I\subset\R^D$. 
\begin{figure}[ht]
	\begin{center}
		\includegraphics[scale=0.28]{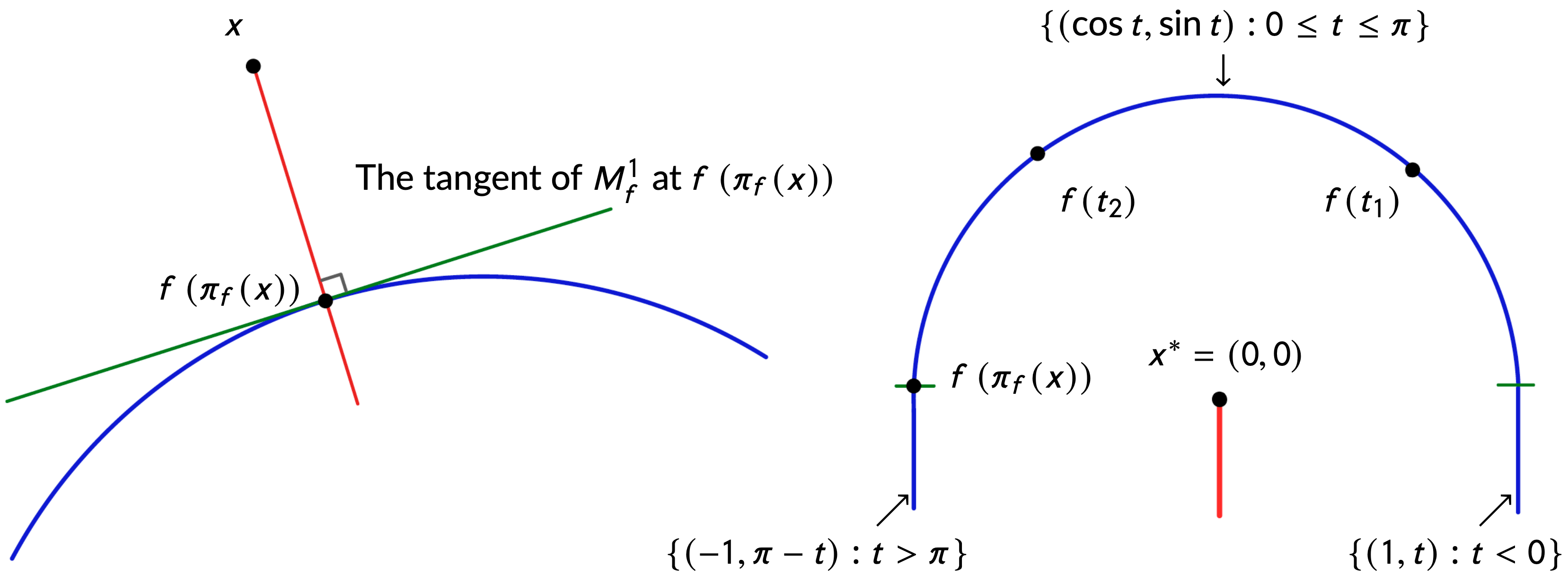}
	\end{center}
	\caption{The left panel illustrates projection indices for $d=1$. In the right panel, $x^*$ is at the center of a semicircle. $\mathcal{A}_f(x^*)=[0,\pi]$ is compact, $\pi_f(x^*)=\pi$ and $\Vert x^*-f(t_1)\Vert_{\R^D}$$=\Vert x^*-f(t_2)\Vert_{\R^D}$$=\Vert x^*-f\left(\pi_f(x^*)\right)\Vert_{\R^D}=$ $dist(x^*,f)$ with $t_1\neq t_2$. All the points in $\{(0,y):y\le 0\}$ (the red line) are ambiguity points.}\label{Projection}
\end{figure}

The projection index $\pi_f$ in (\ref{eq:projection}) is well-defined under HS conditions when $d=1$. We generalize $\pi_f$ for all intrinsic dimensions $d$ under a less stringent condition. Intuitively, the projection index of $x$ to $M^d_f$ is a parameter $t$ such that $f(t)$ is closest to $x$ (left panel of Figure \ref{Projection}). However, there might be more than one $t$ such that $\Vert x-f(t)\Vert_{\R^D}=\inf_{t'\in\R^d}\Vert x-f(t')\Vert_{\R^D}=:dist(x,f)$, resulting in ambiguity in choosing $t$ as shown in the right panel of Figure \ref{Projection}. To remove this ambiguity, we introduce the following function space.
\begin{align*}
C_\infty\left(\R^d\rightarrow\R^D\right)=\left\{f\in C\left(\R^d\rightarrow\R^D\right): \lim_{\Vert t\Vert_{\R^d}\rightarrow\infty} \left\Vert f(t) \right\Vert_{\R^D}=\infty\right\},
\end{align*} 
In applications, this function space is not restrictive. Since a data set with finite sample size is always bounded, we are concerned with fitting functions in that bounded domain. Therefore the behavior of a $C_\infty$ map as $\Vert t\Vert_{\R^d}\rightarrow\infty$ does not limit the applications of this framework. This approach is similar to focusing on the segment of a simple linear regression line within the range of an observed independent variable, even though the fitted line is unbounded. Based on these notations, we have the following theorem.
\begin{theorem}\label{fundation of projection index}
	If $f\in C_\infty(\R^d\rightarrow\R^D)$, then the $d$-dimensional set $$\mathcal{A}_f(x)=\left\{t\in\R^d:\left\Vert x-f(t)\right\Vert_{\R^D}=dist(x,f)\right\}$$ 
	is nonempty and compact for all $x\in\R^D$. 
\end{theorem}
\noindent The proof of Theorem \ref{fundation of projection index} is in the Appendix. The condition $f\in C_\infty(\R^d\rightarrow\R^D)$ is necessary for Theorem \ref{fundation of projection index} to hold as shown by the following example. Let $f(t)=(\frac{e^t}{1+e^t},0)^T\notin C_\infty(\R^1\rightarrow\R^2)$ and $x=(-1,0)^T$, then $\inf_{t\in\R}\Vert f(t)-x\Vert_{\R^2}=\inf_{t\in\R}\vert\frac{e^t}{1+e^t}+1\vert=1$. However, $\vert\frac{e^t}{1+e^t}+1\vert>1$ for all $t$, then we have $\mathcal{A}_f(x)=\emptyset$. We propose a generalized definition of $\pi_f$ as follows. Theorem \ref{fundation of projection index} guarantees that generalized $\pi_f$ is well-defined. 
\begin{definition}\label{Def 3}
	Let $f\in C_\infty(\R^d\rightarrow\R^D)$. (i) If $\mathcal{A}_f(x)$ contains more than one element, then $x$ is called an ambiguity point of $f$. (ii) For any $x\in\R^D$, the projection index $\pi_f(x)$ is defined as $\pi_f(x)=(t^*_1, t^*_2,\cdots, t^*_d)$, where
	\begin{align}\label{eq:projection1}
	    \notag & t^*_1=\max\left\{t_1:(t_1, t_2,\cdots, t_d)\in \mathcal{A}_f(x)\right\},\\ 
	& t^*_j=\max\left\{ t_j: (t^*_1, \cdots, t^*_{j-1}, t_j, \cdots, t_d)\in\mathcal{A}_f(x)\right\},\ \ \ j=2,3,\cdots,d. 
	\end{align}
\end{definition}
\noindent The compactness of $\mathcal{A}_f(x)$ implies $\pi_f(x)\in\mathcal{A}_f(x)$. The projection index in (\ref{eq:projection1}) is a generalization of the projection index defined in (\ref{eq:projection}). Therefore, we use the same notation $\pi_f$ to denote both projection indices. In defining the projection index (\ref{eq:projection1}), we replace the restrictive HS conditions with the less stringent condition $f\in C_\infty(\mathbb{R}^d\rightarrow\mathbb{R}^D)$ and allow $d$ to be any positive integer. When $d=1$, if $f\in C_\infty$ and $f$ satisfies HS conditions, the projection index in (\ref{eq:projection1}) is the same as that in (\ref{eq:projection}). The following theorem, whose proof is in Appendix, implies that $\pi_f(X)$ is a random $d$-vector if $X$ is a random $D$-vector.
\begin{theorem}\label{measurability of projection indices}
	If $f\in C_\infty(\R^d\rightarrow\R^D)$, then (i) $\pi_f$ is measurable, and (ii) $\{\pi_f(x):x\in B\}$ is bounded when $B\subset\R^D$ is compact. Furthermore, if $f$ has no ambiguity point in an open set $U\subset\R^D$ and is a homeomorphism, then (iii) $\pi_f$ is continuous on $U$.
\end{theorem}


\section{Principal Manifolds}\label{definitions of principal manifolds}

For a random $D$-vector $X$, its first $d$ linear principal components are equivalent to the $d$-dimensional hyperplane defined by $\arg\min_{L\in\mathscr{L}}\mathbb{E}\left\Vert X-\Pi_{L}(X)\right\Vert_{\R^D}^2$, where $\mathscr{L}$ is the collection of all $d$-dimensional hyperplanes in $\R^D$ and $\Pi_{L}(X)$ is the projection of $X$ to the hyperplane $L$. We propose a principal manifold framework generalizing PCA, replacing $\mathscr{L}$ with a Sobolev space. In the sequel, all derivatives are \textit{generalized derivatives} defined on $\mathscr{D}'(\R^d)$, where $\mathscr{D}'(\R^d)$ is the collection of generalized functions on $\R^d$. Definitions of $\mathscr{D}'(\R^d)$ and generalized derivatives are provided in Chapter 6 of \cite{rudin1991functional}. The only exception is that the derivatives referred to in the definition of the function space $C^k(\R^d\rightarrow\R^D)$ are still understood in the classical sense. Generalized derivatives, also called ``derivatives of distributions," apply to all locally integrable functions that may not be classically differentiable. They free us from smoothness assumptions in theoretical arguments. 

\subsection{Sobolev Spaces}\label{section: Sobolev Spaces}

We introduce the following function spaces of Sobolev type.
\begin{align*}
& \nabla^{-\otimes 2}L^2\left(\R^d\right)=\left\{f\in\mathscr{D}'(\R^d):\Vert\nabla^{\otimes 2} f\Vert_{\R^{d\times d}}\in L^2(\R^d)\right\},\\ 
& H^2\left(\R^d\right)=\left\{f\in L^2(\R^d):\left\Vert\nabla f\right\Vert_{\R^d},\Vert\nabla^{\otimes 2} f\Vert_{\R^{d\times d}}\in L^2(\R^d) \right\},\\
& \nabla^{-\otimes 2}L^2\left(\Omega\right)=\left\{f\vert_\Omega: f\in \nabla^{-\otimes 2}L^2(\R^d)\right\},\\
& H^2\left(\Omega\right)=\left\{f\vert_\Omega: f\in H^2(\R^d)\right\},
\end{align*}
where $\Omega$ is any open subset of $\R^d$ with a smooth boundary, $f\vert_\Omega$ denotes the restriction of $f$ to $\Omega$, and 
\begin{itemize}
    \item $\nabla f$ denotes the vector-valued function $t\mapsto (\frac{\partial f}{\partial t_1}(t), \frac{\partial f}{\partial t_2}(t), \cdots, \frac{\partial f}{\partial t_d}(t))^T=:\nabla f(t)$, i.e., the gradient of $f$, and $\Vert\nabla f\Vert_{\R^d}$ denotes the scalar-valued function $t\mapsto\Vert \nabla f(t)\Vert_{\R^d}=(\sum_{i=1}^d\vert\frac{\partial f}{\partial t_i}(t)\vert^2)^{1/2}$; 
    
    \item $\nabla^{\otimes2}f$ denotes the matrix-valued function $t\mapsto (\frac{\partial^2f}{\partial t_i \partial t_j}(t))_{1\le i,j\le d}$, i.e., the Hessian matrix of $f$; $\Vert\nabla^{\otimes 2} f\Vert_{\R^{d\times d}}$ denotes the scalar-valued function $t\mapsto \Vert\nabla^{\otimes 2} f(t)\Vert_{\R^{d\times d}}=(\sum_{i,j=1}^d\vert\frac{\partial^2f}{\partial t_i\partial t_j}(t)\vert^2)^{1/2}$.
    
    \item $\Vert\nabla f\Vert_{\R^d}\in L^2(\R^d)$ and $\Vert\nabla^{\otimes2} f\Vert_{\R^d}\in L^2(\R^d)$ denote $\Vert\nabla f\Vert^2_{L^2(\R^d)}=\int_{\R^d} \sum_{i=1}^d\vert\frac{\partial f}{\partial t_i}(t)\vert^2 dt<\infty$ and $\Vert\nabla^{\otimes2} f\Vert^2_{L^2(\R^d)}=\int_{\R^d} \sum_{i,j=1}^d\vert\frac{\partial^2f}{\partial t_i\partial t_j}(t)\vert^2 dt<\infty$, respectively.
    
    \item $\nabla^{\otimes2}$ is an abbreviation of the tensor product $\nabla\otimes\nabla=(\frac{\partial^2}{\partial t_i \partial t_j})_{i\le i,j \le d}$.
\end{itemize}
Section 1.5 of \cite{duchon1977splines} implies the following result. 
\begin{lemma}\label{space equivalence}
	If $\Omega$ is bounded, then $\nabla^{-\otimes 2}L^2(\Omega)=H^2(\Omega)$.
\end{lemma}
\noindent If two functions are equal to each other almost everywhere with respect to the Lebesgue measure on $\R^d$, we identify them as the same function. Then we have the following regularity theorem.
\begin{theorem}\label{Sobolev embedding}
	$\nabla^{-\otimes 2}L^2(\R^d)\subset C^k(\R^d)$, for $k<2-\frac{d}{2}$.
\end{theorem}
\begin{proof}
	Let $f\in\nabla^{-\otimes 2}L^2(\R^d)$ and $\Omega$ be any open ball in $\R^d$, Lemma \ref{space equivalence} implies $f\vert_\Omega\in H^2(\Omega)$. A Sobolev embedding theorem (Theorem 7.25 of \cite{rudin1991functional}) implies $f\vert_\Omega\in C^{k}(\Omega)$ for $k<2-\frac{d}{2}$. Then the result follows as $\Omega$ is arbitrary.
\end{proof}
\noindent For simplicity, define $\nabla^{-\otimes 2}L^2(\R^d\rightarrow\R^D)=\{f(t)=(f_1(t), f_2(t), \cdots, f_D(t))^T: f_l\in \nabla^{-\otimes 2}L^2(\R^d)\mbox{ for all }l=1,2,\cdots, D\}$, and $C_{\infty}\bigcap \nabla^{-\otimes 2}L^2=C_{\infty}\bigcap \nabla^{-\otimes 2}L^2(\R^d\rightarrow\R^D)=C_{\infty}(\R^d\rightarrow\R^D)\bigcap\nabla^{-\otimes 2}L^2(\R^d\rightarrow\R^D)$. 

\subsection{Definition of Principal Manifolds}

As mentioned in Section \ref{introduction}, a problem of the model in \cite{kegl2000learning}, which is equivalent to (\ref{Smola et al}) with $\Vert Pf\Vert_{\mathscr{H}}=\Vert f'\Vert_{L^2}$, is that the fitted $f$ is not necessarily differentiable exactly everywhere. The main reason for this limitation is that the regularization from $\Vert f'\Vert^2_{L^2}$ may not provide enough penalty on the non-smoothness of $f$. We propose principal manifolds with higher regularity by replacing $f'$ with the second derivative $f''$. When $d=1$ and $f$ is arc-length parameterized, $\Vert f''(t)\Vert_{\R^D}$ is the curvature of $f$ at $t$ (Chapter 1.5 of \cite{do2016differential}). Additionally, $\int \Vert f''(t)\Vert_{\R^D}^2 dt$ is usually called the ``total squared curvature" of $f$ (e.g., \cite{enomoto2013total}) and coincides with the penalty term of the cubic smoothing spline framework. For surfaces $f(t_1, t_2)=(f_1(t_1, t_2),\cdots,f_D(t_1, t_2))^T$, the penalty $\sum_{l=1}^D\int_{\mathbb{R}^d} \sum_{i,j=1}^2\vert\frac{\partial^2 f_l}{\partial t_i \partial t_j}(t)\vert^2dt$ in the thin plate spline framework measures the roughness of surfaces $M_f^2$ (see \cite{koenker2004penalized}). For a general intrinsic dimension $d\ge 1$ and the map $f(t)=(f_1(t), f_2(t),\cdots, f_D(t))^T$ with $t\in\R^d$, we generalize these two penalty terms to $\Vert\nabla^{\otimes 2}f \Vert_{L^2(\mathbb{R}^d)}^2=\sum_{l=1}^D \Vert \nabla^{\otimes 2}f_l \Vert^2_{L^2(\R^d)} = \sum_{l=1}^D\int_{\mathbb{R}^d} \sum_{i,j=1}^d\vert\frac{\partial^2 f_l}{\partial t_i \partial t_j}(t)\vert^2dt$. In the sequel, we apply $\Vert\nabla^{\otimes 2}f \Vert_{L^2(\mathbb{R}^d)}^2$ to measure the roughness of $M_f^d$. We broaden the use of the name ``total squared curvature" defined for curves and call its high-dimensional generalization $\Vert\nabla^{\otimes 2}f \Vert_{L^2(\mathbb{R}^d)}^2$ \textit{total squared curvature} as well. The tolerance of large total squared curvature increases the complexity and decreases the stability of fitted manifolds. Therefore, we penalize fitted manifolds with large total squared curvature. These considerations motivate us to define principal manifolds as follows.
\begin{definition}\label{Def 4}
	Let $X$ be a random $D$-vector associated with the probability measure or density function $\mathbb{P}$ such that $X$ has compact support $\supp(\mathbb{P})$ and finite second moments. Let $f,g\in C_\infty\bigcap \nabla^{-\otimes 2}L^2(\R^d\rightarrow\R^D)$ and $\lambda\in[0,\infty]$, we define the following functionals
	\begin{align}\label{PMSEF}
	\mathcal{K}_{\lambda,\mathbb{P}}(f,g)=\mathbb{E}\left\Vert X-f\left(\pi_g(X)\right)\right\Vert^2_{\mathbb{R}^D}+\lambda\left\Vert\nabla^{\otimes 2}f\right\Vert_{L^2(\mathbb{R}^d)}^2, \ \ \ \ \mathcal{K}_{\lambda,\mathbb{P}}(f)=\mathcal{K}_{\lambda,\mathbb{P}}(f,f).
	\end{align}
	A manifold $M_{f^*}^d$ determined by $f^*$ is called a principal manifold for $X$ (or $\mathbb{P}$) with the tuning parameter $\lambda$ if 
	\begin{align}\label{def: PM}
	& f^*_\lambda=\arg\min_{f\in\mathscr{F}(\mathbb{P})}\mathcal{K}_{\lambda,\mathbb{P}}(f),\\ 
	\notag & \mbox{ where }\mathscr{F}(\mathbb{P})=\left\{f\in C_\infty\cap \nabla^{-\otimes 2}L^2(\R^d\rightarrow\R^D):\sup_{x\in\supp(\mathbb{P})}\left\Vert\pi_f(x)\right\Vert_{\R^d}=1\right\}.
	\end{align}
\end{definition}
\noindent Since $\lambda$ above is allowed to be $\infty$, we adopt the convention $\infty\times0=0$ as $\lim_{\lambda\rightarrow\infty}(\lambda\times0)=0$. Then $\mathcal{K}_{\infty,\mathbb{P}}(f)<\infty$ only if $\Vert\nabla^{\otimes 2} f\Vert_{L^2(\mathbb{R}^d)}=0$. Suppose $f^*_\lambda$ is derived, then $\pi_{f^*_\lambda}(X)$ gives a $d$-dimensional parameterization of $X$. Theorem \ref{measurability of projection indices} (iii) implies the continuity of $\pi_{f^*_\lambda}(X)$ under some conditions on $f^*_\lambda$. The constraint $\sup\{\left\Vert\pi_f(x)\right\Vert_{\R^d}: x\in\supp(\mathbb{P})\}=1$ is well-defined (see Theorem \ref{measurability of projection indices} (ii)) and restricts the parameterizations $\{\pi_f(x):x\in\supp(\mathbb{P})\}$ exactly in the unit ball $\{t\in\R^d:\Vert t\Vert_{\R^d}\le1\}$. This restriction is an analog of the arc-length parameterization in the scenario $d=1$. Additionally, Theorem \ref{Sobolev embedding} implies the regularity of principal manifolds, i.e., $f^*_\lambda\in C^{k}(\R^d\rightarrow\R^D)$ for $k<\max\{2-\frac{d}{2},1\}$. The definition above is a special case of the regularized principal manifolds (\ref{Smola et al}) defined by \cite{smola2001regularized}, where $P=\nabla^{\otimes 2}$, $\mathscr{H}=L^2(\mathbb{R}^d)$, and $\mathscr{F}=\mathscr{F}(\mathbb{P})$. 

Motivated by the ``projection-adaptation'' algorithm in Section 5 of \cite{smola2001regularized}, we apply its iterative fashion to estimate principal manifolds. Specifically, we estimate $f^*_\lambda=\arg\min_{f\in\mathscr{F}(\mathbb{P})}\mathcal{K}_{\lambda,\mathbb{P}}(f)$ using 
\begin{align}\label{main iteration}
f_{(n)}=\arg\min_{f}\left\{\mathcal{K}_{\lambda,\mathbb{P}}\left(f,f_{(n-1)}\right): f\in C_\infty\cap \nabla^{-\otimes 2}L^2\left(\R^d\rightarrow\R^D\right) \right\},\ \ \ n=1,2,\cdots,\ \ \lambda\ge0.
\end{align}
Suppose (\ref{main iteration}) stops when $n=n^*$, and we obtain $f_{(n^*)}\in C_\infty\bigcap\nabla^{-\otimes2}L^2(\R^d\rightarrow\R^D)$. Then an estimate of $\arg\min_{f\in\mathscr{F}(\mathbb{P})}\mathcal{K}_{\lambda,\mathbb{P}}(f)$ is given by $\widehat{f^*_\lambda}(t)=f_{(n^*)}(\kappa t)$ with $\kappa=\sup\{\Vert \pi_{f_{(n^*)}}(x)\Vert_{\R^d}:x\in\supp(\mathbb{P})\}<\infty$ (see Theorem \ref{measurability of projection indices} (ii)), and $\widehat{f^*_\lambda}\in\mathscr{F}(\mathbb{P})$. Computing $\pi_{f_{(n)}}(X)$ in $\mathcal{K}_{\lambda,\mathbb{P}}(f,f_{(n)})$ implicitly corresponds to the ``projection" step discussed in Section \ref{geometry}, where our results guarantee that $\pi_{f_{(n)}}(X)$ is well-defined. Minimizing $\mathcal{K}_{\lambda,\mathbb{P}}(f,f_{(n)})$ with respect to $f$ corresponds to the ``adaptation" step. Since $f_{(n)} \in C_\infty\bigcap \nabla^{-\otimes 2}L^2$ for all $n$, Theorem \ref{Sobolev embedding} guarantees the regularity of $f_{(n)}$ for all $n$. The iteration (\ref{main iteration}) usually approximates local minima. Hence, successful implementation of the iteration depends on the choice of the starting values. An initialization strategy with a partial implementation of ISOMAP is illustrated in Section \ref{simulations}. 

\subsection{Two Special Cases}

In this subsection, we discuss two extreme  special cases of the tuning parameter $\lambda$: $\lambda = \infty$ and $\lambda = 0$. We show that these two cases imply linear PCA and the HS principal curve algorithm, respectively. Besides, we discuss potential issues that may arise when using these two extreme cases in applications, leading to the consideration of other values of $\lambda$ in our proposed framework. The following theorem establishes the fact that $\lambda=\infty$ implies linear PCA.
\begin{theorem}\label{principal manifolds and PCA}
Suppose $X$ is a random $D$-vector with finite second moments, $\mathbf{v}_1,\mathbf{v}_2,\cdots,\mathbf{v}_D$ and $e_1, e_2, \cdots, e_D$ are eigenvectors and eigenvalues of the covariance matrix of $X$, respectively. $\mathbf{v}_i$ corresponds to $e_i$ and $e_1\ge\cdots\ge e_d>e_{d+1}\ge\cdots\ge e_D$. Then the principal manifold for $X$ with tuning parameter $\lambda=\infty$ is the linear manifold $\left\{\mathbb{E}X+\sum_{i=1}^d\alpha_i\mathbf{v}_i: \alpha_i\in\R^1\right\}$.
\end{theorem}
\noindent Its proof is in the Appendix. Theorem \ref{principal manifolds and PCA} implies that a large $\lambda$ shrinks principal manifolds towards PCA. However, if the underlying manifolds are nonlinear, the linear manifolds with zero curvature are not satisfactory estimators.

When $\lambda=0$, the estimation of principal manifolds may result in overfitting. For example, let $\mathbb{P}$ be the empirical distribution $\frac{1}{I}\sum_{i=1}^I\delta_{x_i}$ for the data $\{x_i\}_{i=1}^I$. For any $f\in C^\infty(\R^d\rightarrow\R^D)$ satisfying $\{x_i\}_{i=1}^I\subset M_f^d$, i.e., $f$ passes through every data point, $\sup_{i}\Vert\pi_f(x_i)\Vert_{\R^d}=1$, and $f(t)$ is equal to an affine function when $\Vert t\Vert_{\R^d}>M$ for a sufficiently large $M>0$, we have $\mathcal{K}_{0,\mathbb{P}}(f)=0$. Then $M_f^d$ is a principal manifold for $\mathbb{P}$ and overfits the data $\{x_i\}_{i=1}^I$. Additionally, $\lambda=0$ results in self-consistency (\ref{self consistency}), potentially resulting in the saddle issue discussed in Section 1. Specifically, the HS principal curve algorithm is of the following form.
\begin{equation}\label{HS iteration}
f_{(n)}=\mathcal{T}_{HS}f_{(n-1)},\ \  \mbox{ where } \mathcal{T}_{HS} f_{(n-1)}(t)=\mathbb{E}\left(X\big\vert\pi_{f_{(n-1)}}(X)=t\right), \ \ n=1, 2,\cdots.
\end{equation}
If $f_{(n)}$ converges to $f$, then (\ref{HS iteration}) implies (\ref{self consistency}). The following theorem implies that (\ref{HS iteration}) is a special case of (\ref{main iteration}) with $\lambda=0$. 
\begin{theorem} If $f_{(n-1)}$ and $\mathcal{T}_{HS} f_{(n-1)}\in C_\infty\bigcap \nabla^{-\otimes 2}L^2$, then $$\mathcal{T}_{HS} f_{(n-1)}=\arg\min_f \left\{\mathcal{K}_{0,\mathbb{P}}\left(f,f_{(n-1)}\right): f\in C_\infty\bigcap \nabla^{-\otimes 2}L^2 \right\}.$$ 
\end{theorem}
\begin{proof} Let $\mathscr{M}$ be the collection of measurable $\R^d\rightarrow\R^D$ maps. We have $\inf\{\mathcal{K}_{0,\mathbb{P}}(f, f_{(n-1)}):f\in C_\infty\bigcap \nabla^{-\otimes 2}L^2\}\ge\inf\{\mathbb{E}\Vert X-f(\pi_{f_{(n-1)}}(X))\Vert_{\R^D}^2:f\in\mathscr{M}\}=\mathbb{E}\Vert X-\mathbb{E}(X\vert \pi_{f_{(n-1)}}(X))\Vert_{\R^D}^2=\mathbb{E}\Vert X-\mathcal{T}_{HS} f_{(n)}(\pi_{f_{(n)}}(X))\Vert_{\R^D}^2$. Then $\mathcal{T}_{HS} f_{(n-1)}\in C_\infty\bigcap \nabla^{-\otimes 2}L^2(\R^d\rightarrow\R^D)$ implies the desired result.
\end{proof}

It is important to consider the case when $\lambda$ is positive but very close to 0. In this case, the total squared curvature $\Vert\nabla^{\otimes2}f^*_\lambda\Vert^2_{L^2(\R^d)}$ of $f^*_\lambda=\arg\min_{f\in\mathscr{F}(\mathbb{P})}\mathcal{K}_{\lambda, \mathbb{P}}(f)$ may be unreasonably large. The large value of the total squared curvature may result in numerical issues when performing optimization procedures. When $\lambda>0$, similar to the penalty of the ridge regression,  $\Vert\nabla^{\otimes2}f^*_\lambda\Vert^2_{L^2(\R^d)}$ is bounded by a positive $U_\lambda$, a number that depends on $\lambda$. As $\lambda$ decreases, the bound $U_\lambda$ tends to increase. Therefore, in some cases, we may have $\limsup_{\lambda\rightarrow0}\Vert\nabla^{\otimes2}f^*_\lambda\Vert^2_{L^2(\R^d)}=\infty$. We call this phenomenon ``curvature singularity at $\lambda=0$'' throughout this paper. If the curvature singularity at $\lambda=0$ holds, a fitted manifold $M^d_{f^*_\lambda}$ may have high roughness and be unstable for sufficiently small values of $\lambda$. Results in this subsection imply that $\lambda=\infty$ may mask the underlying curvature of the function, $\lambda=0$ may result in overfitting or the saddle issue, and small values of $\lambda$ may result in the curvature singularity at $\lambda=0$.

\section{A New Approach to Principal Manifold Estimation}\label{Joints of a random vector}

We present the outline of our proposed novel principal manifold estimation (PME) algorithm, briefly describing its three main steps - reduction, fitting, and tuning. This algorithm addresses two problems in the framework described in Section \ref{definitions of principal manifolds}: (i) the choice of $\lambda\in(0,\infty)$, and (ii) the reduction of the computational burden in implementing iteration (\ref{main iteration}) and elimination of effects of outliers. In image analysis applications, the size of data is usually very large, resulting in computational burden when applying manifold learning algorithms. One approach to addressing computational burden is subsampling, e.g., \cite{yue2016parameterization}. While leading to faster computation, subsampling may result in removing important sections of a given data set. Our proposed PME algorithm solves these two problems simultaneously. 
\begin{figure}[ht]
	\begin{center}
		\includegraphics[scale=0.305]{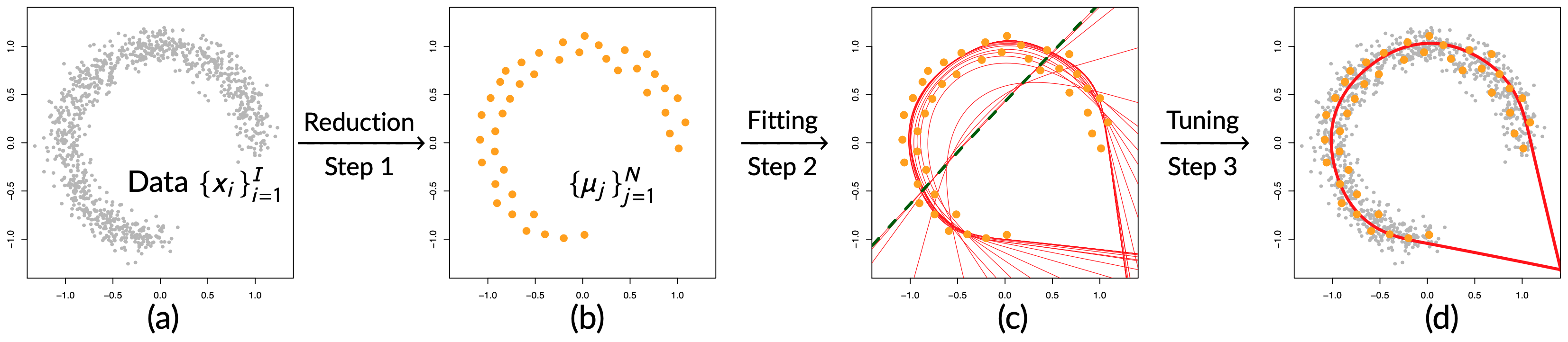}
	\end{center}
	\caption{(a) Data $\{x_i\}_{i=1}^I$ (gray). (b) Each dot (orange) denotes a $\mu_j$. (c) Estimated $\widehat{f}_\lambda$ (red curves) are associated with different $\lambda>0$. The straight line (green and dashed) is associated with $\lambda\rightarrow\infty$ and is an approximation of the first principal component. (d) Using $\{x_i\}_{i=1}^I$, we choose the optimal $\lambda^*$ and draw the corresponding $\widehat{f}_{\lambda^*}$ (red curve). }\label{flow chart}
\end{figure}

A step-by-step visualization of the proposed algorithm is in Figure \ref{flow chart}. In the first stage, the sample size $I$ of the data $\{x_i\}_{i=1}^I$ from $\mathbb{P}$ is reduced to obtain a collection of points $\{\mu_j\}_{j=1}^N$ with a smaller size $N$, where each $\mu_j$ is associated with a weight $\theta_j$ satisfying $\sum_{j=1}^N\theta_j=1$, such that $\{\mu_j\}_{j=1}^N$ preserve the geometric features of the underlying manifold and are less noisy than the original sample $\{x_i\}_{i=1}^I$. Then $\mathcal{K}_{\lambda, \widehat{Q}_N}(f)=\sum_{j=1}^N\theta_j\Vert \mu_j - f(\pi_f(\mu_j))\Vert_{\mathbb{R}^D}^2+\lambda\Vert \nabla^{\otimes 2} f\Vert^2_{L^2(\R^d)}$ approximates $\mathcal{K}_{\lambda, \mathbb{P}}(f)\approx\frac{1}{I}\sum_{i}^I\Vert x_i - f(\pi_f(x_i))\Vert_{\mathbb{R}^D}^2+\lambda\Vert\nabla^{\otimes2}f\Vert^2_{L^2(\R^d)}$ (see Theorem \ref{thm: penalized MSD approximation} in Section \ref{section: HDMDE Algorithm}), where $\mathcal{K}_{\lambda, \widehat{Q}_N}(f)$ is the functional in (\ref{PMSEF}) associated with the probability measure $\widehat{Q}_N=\sum_{j=1}^N\theta_j\delta_{\mu_j}$. This stage results in the reduction of computational burden and elimination of effects of outliers (see Figure \ref{fig:skeleton} (c) and (d)). In the second stage of this approach, for a preselected set of tuning parameters $\lambda>0$, we estimate $\widehat{f}_\lambda=\arg\min_{f}\mathcal{K}_{\lambda, \widehat{Q}_N}(f)$. The collection of estimated functions $\{\widehat{f}_\lambda\}_{\lambda>0}$ prevents the curvature singularity at $\lambda=0$. Finally, in the third stage, we choose an optimal tuning parameter $\lambda^*$ that preserves the geometric structure in the data while avoiding overfitting $\{\mu_j\}_{j=1}^N$. The following theorem implies that $\{\widehat{f}_\lambda\}_{\lambda>0}$ prevents the curvature singularity at $\lambda=0$, i.e, $\limsup_{\lambda\rightarrow0}\Vert\nabla^{\otimes2}\widehat{f}_\lambda\Vert^2_{L^2(\R^d)}<\infty$.
\begin{theorem}\label{bending energy upper bound}
    For all $\lambda>0$, let $\widehat{f}_\lambda=\arg\min_{f}\{\mathcal{K}_{\lambda,\widehat{Q}_N}(f): f\in \mathscr{F}(\widehat{Q}_N) \}$ with $\widehat{Q}_N=\sum_{j=1}^N\theta_j\delta_{\mu_j}$. Then we have the bound
    \begin{align*}
        & \sup\left\{\Vert\nabla^{\otimes2}\widehat{f}_\lambda\Vert_{L^2(\R^d)}^2: \lambda>0 \right\}\\
        & \le\inf \left\{\Vert\nabla^{\otimes2}f \Vert^2_{L^2(\R^d)}: f\in \mathscr{F}(\widehat{Q}_N)\mbox{, and }f\left(\pi_f(\mu_j)\right)=\mu_j\mbox{ for }j=1,2,\cdots,N\right\}=:U_N<\infty.
    \end{align*}
\end{theorem}
\begin{proof}
$\mathcal{W}_N=\{f\in \mathscr{F}(\widehat{Q}_N): f\left(\pi_f(\mu_j)\right)=\mu_j\mbox{ for }j=1,2,\cdots,N\}\ne\emptyset$ implies $U_N<\infty$. If $\sup_{\lambda>0}\Vert\nabla^{\otimes2}\widehat{f}_\lambda\Vert_{L^2(\R^d)}^2>U_N$, there exist $\tilde{\lambda}>0$ and $\tilde{f}\in\mathcal{W}_N$ such that $\Vert\nabla^{\otimes2}\widehat{f}_{\tilde{\lambda}}\Vert_{L^2(\R^d)}^2>\Vert\nabla^{\otimes2}\tilde{f}\Vert_{L^2(\R^d)}^2$. Then $\mathcal{K}_{\tilde{\lambda}, \widehat{Q}_N}(\tilde{f})=\tilde{\lambda}\Vert\nabla^{\otimes2}\tilde{f}\Vert_{L^2(\R^d)}^2<\mathcal{K}_{\tilde{\lambda}, \widehat{Q}_N}(\widehat{f}_{\tilde{\lambda}})$, which contradicts the definition of $\widehat{f}_{\tilde{\lambda}}=\arg\min_{f\in \mathscr{F}(\widehat{Q}_N)}\mathcal{K}_{\tilde{\lambda},\widehat{Q}_N}(f)$.
\end{proof}

\section{Step 1: Data Reduction} \label{section: The Reduction Step of PME}

The reduction step of the PME algorithm is motivated by the data generating mechanism in manifold learning tasks. In this Section, we show that the distribution of the data $\{x_i\}_{i=1}^I$ can be approximated by a probability measure of the form  $\widehat{Q}_N=\sum_{j=1}^N \theta_j \delta_{\mu_j}$, a convex combination of the point masses $\delta_{\mu_j}$. We propose the \textit{high dimensional mixture density estimation} (HDMDE) algorithm for the estimation of $\mu_j$, $\theta_j$, and $N$ in $\widehat{Q}_N$.

\subsection{Motivation and Estimation of Mixture Density Parameters}\label{section: Estimation of Mixture Density Parameters}

In manifold learning, we assume that the $D$-dimensional data $\{x_i\}_{i=1}^I$ are realizations from a $d$-dimensional latent manifold, corrupted by $D$-dimensional noise. Each $x_i$ is generated in two stages - the latent data stage and the noise corruption stage. In the latent data stage, a latent random $D$-vector $T$ is generated from a probability measure $Q^\star$, where $Q^\star$ is supported on a $d$-dimensional manifold. Then in the noise corruption stage, given $T=t$, the data point $x_i$ is generated from a probability density function (PDF) $\psi(\cdot-t)$ with $\int x\psi(x)dx=\pmb{0}\in\R^D$. Then the distribution generating $x_i$ is the PDF $p(x)=\psi* Q^\star(x)=\int \psi(x-t)Q^\star(dt)$, where $*$ denotes the convolution operation. We may estimate the latent probability measure $Q^\star$ by maximizing the nonparametric likelihood $\mathcal{L}(Q)=\prod_{i=1}^I \psi*Q(x_i)$ in $Q\in\mathscr{Q}$, where $\mathscr{Q}$ denotes the collection of probability measures supported on $d$-dimensional manifolds. Theorem 3.1 of \cite{lindsay1983geometry} implies that there exists a unique probability measure of the form $\widehat{Q}_N=\sum_{j=1}^N\theta_j\delta_{\mu_j}$ with $N\le I$ achieving $\sup_{Q\in\mathscr{Q}}\mathcal{L}(Q)$. For example, in Figure \ref{flow chart} (d), the gray dots denote data $x_i$, the red curve denotes the support of the latent variable $T$ (the probability measure $Q^\star$), and the large orange dots illustrate the point masses $\delta_{\mu_j}$ in $\widehat{Q}_N$.

Substituting the true $Q^\star$ with the maximizer $\widehat{Q}_N$, we estimate the distribution generating $x_i$ by the PDF $\psi*\widehat{Q}_N(x)=\sum_{j=1}^N\theta_j\psi(x-\mu_j)$. To estimate $\mu_j$, $\theta_j$, and $N$, we use a mixture density estimation approach. For any $\sigma>0$, denote $\psi_{\sigma}(x)=\frac{1}{\sigma^D}\psi\left(\frac{x}{\sigma}\right)$. For any fixed positive integer $N$, let $\{\mu_{j,N}\}_{j=1}^N$ be a collection of points in $\R^D$ depending on $N$, let $\theta_N=(\theta_{1,N},\theta_{2,N},\cdots,\theta_{N,N})^T$ be in the probability simplex $\Theta_N=\{\theta_N:\theta_{j,N}\ge 0,\sum_{j=1}^N\theta_{j,N}=1\}$, and let $\sigma_N$ be a positive number such that $\lim_{N\rightarrow\infty}\sigma_N=0$. We construct the following mixture density
\begin{align}\label{approximating density}
	p_N(x\vert\mathbf{\theta}_N)=\psi_{\sigma_N}*\widehat{Q}_N(x)=\sum_{j=1}^N\theta_{j,N}\psi_{\sigma_N}\left(x-\mu_{j,N}\right),\ \ \ \mbox{where }\widehat{Q}_N=\sum_{j=1}^N \theta_{j,N}\delta_{\mu_{j,N}}.
	\end{align}
We estimate $\mu_{j,N}$, $\theta_N$, $\sigma_N$, and $N$ such that $p_N(x\vert\mathbf{\theta}_N)$ approximates the true PDF $p(x)$ generating data $\{x_i\}_{i=1}^I$.

Assuming that the number of mixture components $N$ is fixed, various approaches may be implemented for the estimation of mixture parameters  $\mu_{j,N}$, $\theta_{j,N}$, and $\sigma_N$. A common approach for estimating these parameters is based on the EM algorithm (\cite{dempster1977maximum}). However, this approach can be too computationally intensive in our setting. We propose a high-dimensional generalization of the mixture density estimation algorithm proposed by \cite{eloyan2011smooth}, where the estimation of $\mu_{j,N}$ and $\sigma_N$ is performed in a computationally efficient manner for a given $N$ and the estimation of the mixture weights $\theta_{j,N}$ is then conducted using the EM algorithm. \\
\textbf{Estimation of $\mu_{j,N}$:} Partition $\{x_i\}_{i=1}^I$ to $N$ clusters by \textit{k-means clustering}. Let $\{\mu_{j,N}\}_{j=1}^N$ be the centers of the $N$ clusters. \\
\textbf{Estimation of $\sigma_N$:} Let $\{x_{j,l}\}_{l=1}^{L_j}$ be $x_i$ in $j^{th}$ cluster. We estimate $\sigma_N$ by $$\widehat{\sigma}_N=\left\{\frac{1}{D\times N}\sum_{j=1}^N\left(\frac{1}{L_j}\sum_{l=1}^{L_j}\left\Vert x_{j,l}-\mu_{j,N}\right\Vert_{\R^D}^2\right)\right\}^{1/2}.$$ If $\left\{x_{j,l}\right\}_{l=1}^{L_j}$ are iid $N_D(\mu_{j,N},\sigma_N^2 I_{D\times D})$ for $j=1,2,\cdots,N$, then $\widehat{\sigma}_N^2$ is an unbiased estimator of $\sigma_N^2$.\\
\textbf{Estimation of $\theta_{j,N}$:} Assuming that $\{x_i\}_{i=1}^I$ is a random sample from the PDF $$p_N(x\vert\theta_N)=\sum_{j=1}^N\theta_{j,N}\psi_{\sigma_N}\left(x-\mu_{j,N}\right)\approx p(x),$$ we estimate $\theta_{j,N}$ by likelihood maximization. In practice, the sample mean is used as an unbiased estimate of $\mathbb{E}_{p_N}(X\vert\theta_N)$. Therefore, we apply the constraint $\int_{\R^D} x p_N(x\vert\theta_N)dx  = \overline{x}=\frac{1}{I}\sum_{i=1}^I x_i$ and estimate $\theta_{j,N}$ achieving the likelihood maximization using a constrained EM algorithm. The detailed derivation of this procedure is in Appendix. The proposed approach results in the following equation for updating the estimate of $\theta_{j,N}$.
\begin{align}\label{EM iteration 3}
& \theta^{(k)}_{j,N}=\frac{\sum_{i=1}^I w_{ij}\left(\theta_N^{(k-1)}\right)}{\hat{\rho}_1+\hat{\rho}_2^T\mu_{j,N}},\ \ \ j=1,2,\cdots,N \mbox{ and }k=1,2, \cdots, \mbox{ where}\\
\notag & \left(\widehat{\rho}_1,\widehat{\rho}_2\right)=\arg\min_{\rho_1\in\mathbb{R},\rho_2\in\mathbb{R}^D} \left\{\left\vert\sum_{j=1}^N \left(\frac{\sum_{i=1}^Iw_{ij}\left(\theta_N^{(k-1)}\right)}{\rho_1+\rho_2^T\mu_{j,N}} \right)-1\right\vert^2+ \left\Vert\sum_{j=1}^N \left(\frac{\sum_{i=1}^I w_{ij}\left(\theta_N^{(k-1)}\right)}{\rho_1+\rho_2^T\mu_{j,N}} \right)\mu_{j,N}-\overline{x}\right\Vert_{\mathbb{R}^D}^2 \right\},\\
\notag & w_{ij}\left(\theta_N\right)=\frac{\theta_{j,N}\times\psi_{\widehat{\sigma}_N}\left(x_i-\mu_{j, N}\right)}{\sum_{j'=1}^N\theta_{j',N}\times\psi_{\widehat{\sigma}_N}\left(x_i-\mu_{j', N}\right)}.
\end{align}
If $\sup\left\{\vert \theta_{j,N}^{(k^*)}-\theta_{j,N}^{(k^*-1)}\vert:j=1,2,\cdots,N \right\}$ is smaller than a predetermined threshold, then the iteration (\ref{EM iteration 3}) stops, and $\theta_N^{(k^*)}=(\theta_{1,N}^{(k^*)}, \theta_{2,N}^{(k^*)}, \cdots,  \theta_{N,N}^{(k^*)})^T$ is an estimate of $\theta_N$.

\begin{figure}[ht]
	\begin{center}
	    \includegraphics[scale=0.31]{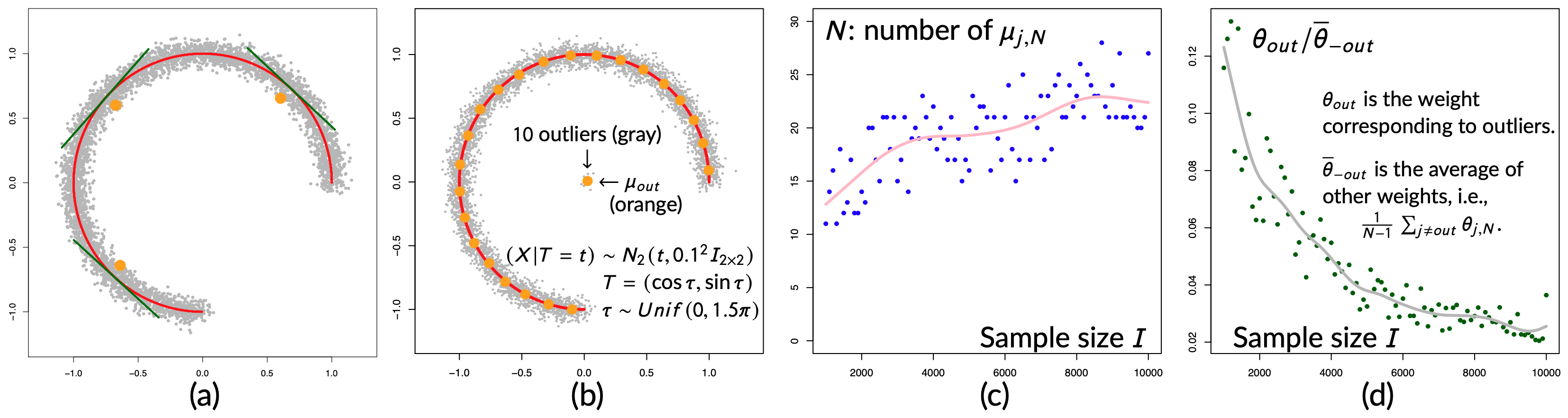}
	\end{center}
	\caption{(a) The red $3/4$ circle is the underlying manifold of interest, the data cloud (gray) is generated from the $X$ in (b), applying $k$-means clustering with $N=3$ results in three large centers (orange); the 3 centers lie completely inside the circle and do not give a good representation of the red circle. This issue stems from the nonzero curvature: For each tangent (green) of the circle, data are unevenly distributed on its two sides. Cluster centers are ``dragged" to the side with more data points. This issue disappears when $N$ is sufficiently large so that $\Vert p_N(\cdot\vert\theta_N)-p\Vert_{L^1(\R^D)}\approx0$. (b) Small dots (gray) are random samples from $X$. The support of $T$ is the solid curve (red). We apply Algorithm \ref{HDMDE algorithm} to these gray dots with input $N_0=10, \alpha=0.05, \epsilon=0.001$. The large dots (orange) denote the estimated $\mu_{j,N}$ in the $\widehat{Q}_N=\sum_{j=1}^N\theta_{j,N}\delta_{\mu_{j,N}}$. (c) Set $N_0=10$, for random samples with size $I$ ranging from $1000$ to $10000$, the estimated $N$ are shown by dots (blue). The curve (pink) shows the trend of $N$ as the sample size $I$ increases. (d) Illustration of the influence of outliers on $\widehat{Q}_N$ as $I$ increases. For each $I$, the influence of outliers is measured by the quantity $\theta_{out}/\overline{\theta}_{-out}$ and shown by a dot (green). The gray curve shows the corresponding trend.}\label{fig:skeleton}
\end{figure}

\subsection{Estimation of the Number of Mixture Components}\label{section: Estimation of the Number of Mixture Components}  

In this subsection, we propose an iterative hypothesis testing procedure to choose the number of mixture components $N$. If $N$ is too small, $\widehat{Q}_N$ in (\ref{approximating density}) may not capture geometric features of data $x_i$. A detailed example of the ``small $N$" issue is presented in Figure \ref{fig:skeleton} (a). On the other hand, an unreasonably large $N$ may result in computational burden and redundant model complexity. Motivated by the following theorem, we choose $N$ by investigating the $L^1$-distance between $p_N(\cdot\vert\theta_N)=\psi*\widehat{Q}_N$ in (\ref{approximating density}) and the PDF $p$ generating data $x_i$. The diameter of a set $U$ in $\R^D$ is defined by $diam(U)=\sup\{\Vert x_1-x_2\Vert_{\R^D}:x_1,x_2\in U\}$.
\begin{theorem}\label{density estimation}
 Suppose $p$ is a PDF with compact support $\supp (p)=\overline{\{ x:p(x)\ne0\}}$ and $p\in L^q(\R^D)$ for some $1\le q<\infty$, $\mathcal{M}_N=\{\mu_{j,N}\}_{j=1}^N\subset \supp (p)$, and $d_N, \sigma_N>0$ for all positive integers $N$. If (i) the triplet $(d_N,\sigma_N,\psi)$ satisfies
 \begin{align}\label{triplet condition}
     \lim_{N\rightarrow\infty}\left(\sup\left\{\left\Vert\psi_{\sigma_N}(\cdot-y)-\psi_{\sigma_N}\right\Vert_{L^q(\R^D)}:\Vert y\Vert_{\R^D}\le d_N\right\}\right)=\lim_{N\rightarrow \infty}\sigma_N=\lim_{N\rightarrow \infty}d_N=0;
 \end{align}
 (ii) there exists a partition of the compact set $\supp(p)$, i.e., $\supp(p)=\bigcup_{j=1}^N A_{j,N}$ with $A_{i,N}\bigcap A_{j,N}=\emptyset$ when $i\ne j$, such that $A_{j,N}\bigcap \mathcal{M}_{N}=\{\mu_{j,N}\}$ and $\sup\{diam(A_{j,N}):j=1,2,\cdots,N\}\le d_N$ for all large positive integers $N$; then there exists a sequence $\{\theta_N\}_N$ with $\theta_N\in\Theta_N$ such that $\lim_{N\rightarrow\infty}\left\Vert p_N\left(\cdot\vert\theta_N\right)-p\right\Vert_{L^q(\R^D)}=0,$ where $p_N(\cdot\vert\theta_N)$ is defined by (\ref{approximating density}).
\end{theorem}
\noindent The proof of Theorem \ref{density estimation} is in Appendix. In applications, observed data are always in a bounded domain. Thus the assumption on $p$ is not restrictive. The triplet satisfying (\ref{triplet condition}) exists,  e.g., $\psi$ is any PDF, $\sigma_N=N^{-\alpha_1}$, and $d_N=N^{-(\alpha_1+\alpha_2)}$, then this triplet satisfies (\ref{triplet condition}) when $q=1$, where $\alpha_1$ and $\alpha_2$ are allowed to be any positive numbers. Herein, we set $\psi$ to be the standard Gaussian kernel $(2\pi)^{-D/2}\exp\{-\Vert x\Vert_{\R^D}^2/2\}$. Condition (ii) essentially requires $\mathcal{M}_N$ to be dense in $\supp (p)$ as $N\rightarrow\infty$. Since we estimate the knots $\mu_{j,N}$ as centers of the $N$ k-means clusters of the data $\{x_i\}_{i=1}^I \sim_{iid} p$, $\mathcal{M}_N=\{\mu_{j,N}\}_{j=1}^N$ tends to be dense in $\supp (p)$ as the number of clusters increases. Therefore, condition (ii) is realistic. Since all PDFs are in $L^1(\R^D)$, we are only interested in the special case of Theorem \ref{density estimation} where $q=1$.

The limit $\lim_{N\rightarrow\infty}\Vert p_N(\cdot\vert\theta_N)-p\Vert_{L^1(\R^D)}=0$ implies $\lim_{N\rightarrow\infty}\Vert p_{N+1}(\cdot\vert\theta_{N+1})- p_{N}(\cdot\vert\theta_{N})\Vert_{L^1(\R^D)}=0.$ If we further assume $p\in L^\infty(\R^D)$, then we have the following limit motivating the proposed method of selecting $N$.
\begin{align*}
  \left\vert\mathbb{E}_p\left\{p_{N+1}\left(X\vert\theta_{N+1}\right)-p_N\left(X\vert\theta_N\right)\right\}\right\vert \le \Vert p\Vert_{L^\infty(\R^D)}\left\Vert p_{N+1}(\cdot\vert\theta_{N+1})- p_{N}(\cdot\vert\theta_{N})\right\Vert_{L^1(\R^D)}\rightarrow0,\ \ \ \ \mbox{as }N\rightarrow\infty,  
\end{align*}
where $X\sim p$. This limit implies $\mathbb{E}_p\left\{p_{N+1}\left(X\vert\theta_{N+1}\right)-p_N\left(X\vert\theta_N\right)\right\}\approx0$ when $N$ is sufficiently large. Therefore, we choose a sufficiently large $N$ by testing the following hypothesis.
\begin{equation}\label{eq:hypothesis}
H_0: \mathbb{E}_p\left\{p_{N+1}(X\vert\theta_{N+1})-p_N(X\vert\theta_{N})\right\} = 0\ \ \ \ \ \mbox{ vs }\ \ \ \ \ H_a: \mathbb{E}_p\left\{p_{N+1}(X\vert\theta_{N+1})-p_N(X\vert\theta_{N})\right\} \neq 0,
\end{equation}
where $\mathbb{E}_p$ is the expectation associated with the PDF $p$. Since $p$ and $\theta_N$ are unknown, we use $\overline{\Delta}_{I,N}=\frac{1}{I}\sum_{i=1}^I \widehat{\Delta}_i$ to test the hypothesis (\ref{eq:hypothesis}), where $\widehat{\theta}_N=\widehat{\theta}_N\left(X_1,X_2,\cdots,X_I\right)$ is an estimator of $\theta_N$ computed from the iid sample $X_i$ and $\widehat{\Delta}_i=p_{N+1}(X_i\vert\widehat{\theta}_{N+1})-p_N(X_i\vert\widehat{\theta}_{N})$. We apply the following theorem to perform an asymptotic test of (\ref{eq:hypothesis}).
\begin{theorem}\label{HT theorem}
	Suppose $\widehat{\theta}_n$ is an estimator of the true $\theta_n\in\Theta_n$, such that $\widehat{\theta}_n=\theta_n+o_p(I^{-1/2})$, where $n\in\{N,N+1\}$, $N$ is fixed, and $\psi\in L^\infty(\R^D)$. Denote $\Delta_i=p_{N+1}(X_i\vert\theta_{N+1})-p_N(X_i\vert\theta_{N})$, $\mu_{\Delta,N}=\mathbb{E}_p\Delta_1$, $\widehat{S}^2_{I,N}=\frac{1}{I}\sum_{i=1}^I\widehat{\Delta}_i^2-( \overline{\Delta}_{I,N})^2$ and $\widehat{S}_{I,N}=\sqrt{\widehat{S}^2_{I,N}}$. Then $\sqrt{I}\frac{\overline{\Delta}_{I,N}-\mu_{\Delta,N}}{\widehat{S}_{I,N}}\rightarrow N(0,1)$ in distribution as $I\rightarrow\infty$.
\end{theorem}
\noindent Theorem \ref{HT theorem} can be derived directly from the central limit theorem and Slutsky's theorem, hence its proof is omitted. Since we are interested in testing the hypothesis $H_0: \mu_{\Delta,N}=0$ as shown in (\ref{eq:hypothesis}), we define the statistic $Z_{I,N} = \sqrt{I}\frac{\overline{\Delta}_{I,N}}{\widehat{S}_{I,N}}$. From Theorem \ref{HT theorem}, under $H_0$, we have $Z_{I,N} \sim N(0,1)$ approximately when $I$ is large. We choose $N$ by $$N=N_\alpha=\inf\left\{N:\vert Z_{I,N}\vert<z_{1-\alpha/2}\mbox{ and }N_0\le N\le I-1\right\},$$ 
where $N_0$ denotes a predetermined lower bound for $N$, $z_{1-\alpha/2}$ is the $1-\alpha/2$ quantile of $N(0,1)$, and $\alpha = 0.05$ is chosen for testing (\ref{eq:hypothesis}). Figure \ref{fig:skeleton} (c) shows that the estimated $N$ tends to be much smaller than $I$. In both the method proposed above and the counterpart in \cite{eloyan2011smooth}, the number of mixture components $N$ is chosen by measuring the dissimilarity between $p_{N+1}(\cdot\vert\theta_{N+1})$ and $p_{N}(\cdot\vert\theta_{N})$. \cite{eloyan2011smooth} applies \textit{Kullback-Leibler divergence} (KLD) while we apply $L^1$-norm. We choose $L^1$-norm because we found in simulations that $L^1$-norm captures the geometric features of $p$ better than KLD.

\subsection{The High-Dimensional Mixture Density Estimation (HDMDE) Algorithm}\label{section: HDMDE Algorithm} 

By iteratively combining the procedures for the estimation of mixture weights $\theta_{j,N}$, mixture element means $\mu_{j,N}$, and the number of mixture components $N$ presented in Sections \ref{section: Estimation of Mixture Density Parameters} and \ref{section: Estimation of the Number of Mixture Components}, we propose the HDMDE algorithm in Algorithm 1 to estimate $\widehat{Q}_N=\sum_{j=1}^N\theta_{j,N}\delta_{\mu_{j,N}}$. In this Section, we use simulation studies to illustrate that HDMDE results in a substantial reduction in computation speed and elimination of the effects of outliers. In addition, we present a proof-of-concept simulation study showing that the density estimated by HDMDE can approximate the true density better than the kernel density estimate (KDE) in terms of minimizing the $L^1$-distance.
\begin{algorithm}[ht]
	\caption{HDMDE}\label{HDMDE algorithm}
	\begin{algorithmic}[1]
		\INPUT (i) Data points $\{x_i\}_{i=1}^I$ in $\R^D$, (ii)  a positive integer $N_0<I-1$, and (iii) $\epsilon, \alpha\in(0,1)$.
		\OUTPUT $N$, $\{\mu_{j,N}\}_{j=1}^{N}$, $\widehat{\theta}_N=(\widehat{\theta}_{1,N},\cdots \widehat{\theta}_{N,N})^T$, and $\sigma_N$. Then we have 
		$$ \widehat{Q}_N=\sum_{j=1}^N\widehat{\theta}_{j,N}\delta_{j,N} \mbox{ and } p_N\left(\cdot\vert\widehat{\theta}_N\right)=\psi_{\sigma_N}*\widehat{Q}_N.$$
		\STATE  $N\leftarrow N_0$ and formally $Z_{I,N}\leftarrow2\times z_{1-\alpha/2}$.
		\STATE\label{mu and sigma step}  Estimate $\mu_{j,N}$ and $\sigma_N$ using the k-means clustering.
		\STATE\label{theta step} Apply the iteration (\ref{EM iteration 3}) and get a sequence $\{\widehat{\theta}_{N}^{(k)}\}_k$. Set $\widehat{\theta}_{N}=\widehat{\theta}_{N}^{(k^*)}$ with 
		$$k^*=\min \left\{k:\sup_j\left\vert\widehat{\theta}_{j,N}^{(k-1)}-\widehat{\theta}_{j,N}^{(k)}\right\vert<\epsilon \right\}.$$
		\STATE\label{hypothesis testing step} Compute $p_N(x_i\vert\widehat{\theta}_{N})=\sum_{j=1}^N \widehat{\theta}_{j,N}\times\psi_{\sigma_N}(x_i-\mu_{j,N})$ for all $i$.
		\WHILE{$\left\vert Z_{I,N}\right\vert\ge z_{1-\alpha/2}$ and $N< I-1$,}
		\STATE $N\leftarrow N+1$, repeat the steps \ref{mu and sigma step}, \ref{theta step}, \ref{hypothesis testing step}, and compute $Z_{I,N}$.
		\ENDWHILE 
	\end{algorithmic}
\end{algorithm}

When the sample size $I$ of data $\{x_i\}_{i=1}^I$ is large, manifold fitting can be computationally expensive. One advantage of using HDMDE in PME is the comparatively small computational burden in estimating $\widehat{f}_\lambda=\arg\min_{f}\mathcal{K}_{\lambda, \widehat{Q}_N}(f)$. We conduct simulation studies to compare the magnitudes of estimated $N$ and sample size $I$ empirically and show that $N$ is much smaller than $I$. Figure \ref{fig:skeleton} (b,c) provides an illustrative comparison between $N$ and $I$ by simulations. For each $I$ ranging from $1000$ to $10000$, we generate $I-10$ points close to a $3/4$ part of a unit circle as presented in Figure \ref{fig:skeleton} (b) and 10 outliers from $N_2(\mathbf{0}, 0.1^2 I_{2\times2})$. We estimate a $\widehat{Q}_N$ by HDMDE for each simulated sample. In Figure \ref{fig:skeleton} (b), we show one simulation example with $I=5000$ by gray points and the estimated $\{\mu_{j,N}\}_{j=1}^N$ by large orange dots. Figure \ref{fig:skeleton} (c) illustrates the estimated $N$ versus $I$ and shows that $N$ are much smaller than $I$.

Another advantage of HDMDE is that the effect of outliers on $\widehat{Q}_N$ is negligible. As a result, when we fit a manifold by $\widehat{f}_\lambda=\arg\min_f\mathcal{K}_{\lambda, \widehat{Q}_N}(f)$, the result is robust to outliers. Specifically, if the node $\mu_{j',N}$ is closer to outliers than to the main part of the data cloud, the associated weight $\theta_{j',N}$ is small. In each of the simulations in Figure \ref{fig:skeleton}, only one node defined as $\mu_{out}$ is located in the outlier cluster, i.e., in $\{x: \Vert x\Vert_{\R^2}<0.3\}$. We denote the weight associated with $\mu_{out}$ by $\theta_{out}$, and denote the average of other weights $\frac{1}{N-1}\sum_{j\ne out}\theta_{j,N}$ by $\overline{\theta}_{-out}$. The ratio $\theta_{out}/\overline{\theta}_{-out}$ measures the influence of $\mu_{out}$ compared to that of other $\mu_{j,N}$. The lower this ratio, the more negligible the effect of outliers on estimation of $\widehat{Q}_N$. Figure \ref{fig:skeleton} (d) shows that $\theta_{out}/\overline{\theta}_{-out}$ is small and decreases drastically as the sample size $I$ increases. Hence, the point $\mu_{out}$ representing $10$ outliers has a negligible effect on $\widehat{Q}_N$ and this effect decreases as $I$ increases. 

An important property of HDMDE is its performance in approximating the true PDF $p$ in terms of minimizing the $L^1$-distance. We compare HDMDE with KDE in terms of approximating $p$ and then use this property of HDMDE to justify the reduction from $\mathcal{K}_{\lambda, p}(f)$ to $\mathcal{K}_{\lambda, \widehat{Q}_N}(f)$ in PME. To apply KDE in a simulation example, we implement the \texttt{R} function \texttt{kde} in the package \texttt{ks} using default parameter values provided in the package. Let $p_{hdmde}(\cdot\vert \pmb{X})$ and $p_{kde}(\cdot\vert \pmb{X})$ denote the PDFs estimated by applying HDMDE and KDE, respectively, to data $\pmb{X}=\{X_i\}_{i=1}^I\sim_{iid}p$. The difference between the performances of HDMDE and KDE is measured by $\Vert p_{kde}(\cdot\vert \pmb{X})-p \Vert_{L^1(\R^D)}-\Vert p_{hdmde}(\cdot\vert \pmb{X})-p \Vert_{L^1(\R^D)}=:\mathcal{J}(\pmb{X})$. Let $p$ be the PDF of the random vector $X$ in Figure \ref{fig:skeleton} (b) (without the $10$ outliers). Using this PDF $p$ as an example, we generate $500$ realizations of $\pmb{X}$ from $p$ with $I=1000$ and estimate the mean $\mathbb{E}\mathcal{J}(\pmb{X})$ and variance $\mathbb{V}\mathcal{J}(\pmb{X})$  using the sample mean and sample variance. Then we compute the Wald $95\%$-confidence interval $\mathbb{E}\mathcal{J}(\pmb{X})\pm1.96\sqrt{\mathbb{V}\mathcal{J}(\pmb{X})}\approx(0.018, 0.130)$. This interval shows that, on average, HDMDE performs better than KDE in the $L^1$-approximation of the density $p$ in this example. Using this property of HDMDE, we provide the following result showing that minimizing $\mathcal{K}_{\lambda, p}(f)$ is approximately equivalent to minimizing $\mathcal{K}_{\lambda, \widehat{Q}_N}(f)$.
\begin{theorem}\label{thm: penalized MSD approximation}
Suppose (i) $p$ and $\psi$ are PDFs with bounded support; (ii) $\{\widehat{Q}_N=\sum_{j=1}^N\theta_{j,N}\delta_{\mu_{j,N}}\}_{N=1}^\infty$ satisfies $\{\mu_{j,N}\}_{j=1}^N\subset \supp (p)$ for all $N$, and $\lim_{N\rightarrow\infty}\Vert\psi_
{\sigma_N}*\widehat{Q}_N - p\Vert_{L^1(\R^d)}=0$ for a sequence $\{\sigma_N\}_{N=1}^\infty$ with $\lim_{N\rightarrow\infty}\sigma_N=0$; (iii) $f\in C_\infty\bigcap \nabla^{-\otimes 2}L^2$ is a homeomorphism and has no ambiguity point in a neighborhood of $\supp (p)$. If there exists $\{\mu_j\}_{j=1}^\infty\subset\R^D$ so that $\lim_{N\rightarrow\infty}\mu_{j,N}=\mu_j$ and $\sum_{j=1}^\infty(\sup_{N': N'\ge j}\theta_{j,N'})<\infty$, we have the limit $\lim_{N\rightarrow\infty}\mathcal{K}_{\lambda, \widehat{Q}_N}(f)=\mathcal{K}_{\lambda, p}(f)$ for $\lambda\in[0,\infty]$.
\end{theorem}
\noindent The proof of Theorem \ref{thm: penalized MSD approximation} is in Appendix. Although the Gaussian kernel $\psi$ does not have a bounded support, most of its mass is in a bounded domain, e.g., the Gaussian kernel in $\R^3$ satisfies $\psi(x)\le10^{-22}$ when $\Vert x\Vert_{\R^3}\ge10$. In Theorem \ref{thm: penalized MSD approximation}, condition (ii) can be implied by Theorem \ref{density estimation}, and condition (iii) is related to Theorem \ref{measurability of projection indices}.

\section{Step 2: Fitting, Step 3: Tuning, Model Complexity Selection}\label{explicit formula}

In this Section, we propose the details of Step 2 (fitting) and Step 3 (tuning) of our proposed PME algorithm illustrated in Figure \ref{flow chart}. To fit $\widehat{f}_\lambda=\arg\min_{f}\mathcal{K}_{\lambda,\widehat{Q}_N}(f)$ in Step 2 of the PME algorithm, we apply the iteration (\ref{main iteration}) with $\mathbb{P}=\widehat{Q}_N$, i.e., $f_{(n+1)}=\arg\min_{f}\left\{\mathcal{K}_{\lambda, \widehat{Q}_N}(f, f_{(n)}): f=(f_1, f_2, \cdots, f_D)^T\in C_\infty\bigcap\nabla^{-\otimes 2}L^2(\R^d\rightarrow\R^D)\right\}$ with
\begin{align}\label{wowowow}
\mathcal{K}_{\lambda, \widehat{Q}_N}(f, f_{(n)})=\sum_{l=1}^D\left\{\sum_{j=1}^N\theta_{j,N}\left\vert\mu_{j,N,l}-f_{l}\left(\pi_{f_{(n)}}\left(\mu_{j,N}\right)\right)\right\vert^2+\lambda\left\Vert \nabla^{\otimes 2} f_{l}\right\Vert^2_{L^2(\R^d)}\right\},
\end{align} 
where $\mu_{j,N,l}$ is the $l^{th}$ component of the $D$-vector $\mu_{j,N}$, and $l$ denotes a vector component index. Define the following notations: (i) If $\nu$ is an even integer, $\eta_\nu(t)=\Vert t\Vert_{\R^d}^\nu\log\left(\Vert t\Vert_{\R^d}\right)$ when $\Vert t\Vert_{\R^d}\ne0$ and $\eta_\nu(t)=0$ when $\Vert t\Vert_{\R^d}=0$; otherwise, $\eta_\nu(t)=\Vert t\Vert_{\R^d}^\nu$. (ii) $Poly_1[t]$ is the linear space of polynomials on $\R^d$ with degree $\le1$ and has a linear basis $\{p_k\}_{k=1}^{d+1}$. The following theorem implies that the minimizer of $\mathcal{K}_{\lambda, \widehat{Q}_N}(\cdot, f_{(n)})$ in $C_\infty\bigcap\nabla^{-\otimes 2}L^2$ is of a spline form.  
\begin{theorem}\label{main functional theorem}
	Suppose $f_{(n)}\in C_\infty(\R^d\rightarrow\R^D)$, $d\le3$, and each polynomial in $Poly_1[t]$ is uniquely determined by its values on $\mathcal{C}=\{\pi_{f_{(n)}}(\mu_{j,N})\}_{j=1}^N$. Then a minimizer of $\mathcal{K}_{\lambda, \widehat{Q}_N}(\cdot, f_{(n)})$ within $C_\infty\bigcap\nabla^{-\otimes 2}L^2(\R^d\rightarrow\R^D)$ is of the following form.
	\begin{align}\label{analytic expression of minimizer}
	f_{(n+1), l}(t)=\sum_{j=1}^N s_{j,l}\times\eta_{4-d}\left(t-\pi_{f_{(n)}}(\mu_{j,N})\right)+\sum_{k=1}^{d+1}\alpha_{k,l}\times p_k(t),\ \ l=1,2,\cdots,D,
	\end{align}
	with the constraint $\sum_{j=1}^N s_{j,l}\times p_k\left(\pi_{f_{(n)}}(\mu_{j,N})\right)=0$ for all $k=1,2,\cdots,d+1$ and $l=1,2,\cdots,D$.
\end{theorem}
\noindent Proof of Theorem \ref{main functional theorem} is in Appendix. The reason for the dimension restriction $d\le3$ is that $\nabla^{-\otimes 2}L^2(\R^d)$ is a reproducing kernel Hilbert space only if $d\le3$ (see \cite{wahba1990spline}, Chapter 2.4). For the purpose of visualization, the intrinsic dimension $d\le3$ is not restrictive. When $d=1$, (\ref{analytic expression of minimizer}) is a cubic smoothing spline. When $d=2$, (\ref{analytic expression of minimizer}) is a thin plate spline. From Theorem \ref{main functional theorem} and the calculation strategy
in Chapter 2 of \cite{wahba1990spline}, it follows that minimizing (\ref{wowowow}) in $C_\infty\bigcap\nabla^{-\otimes 2}L^2$ is equivalent to obtaining the minimizers 
$$\arg\min_{(s_l, \alpha_l)} \left\{ \left\Vert \mathbf{W}^{1/2}\left(\mu_l-\mathbf{E}s_l-\mathbf{T}\alpha_l\right) \right\Vert_{\mathbb{R}^N}^2+\lambda \left\Vert \mathbf{E}^{1/2}s_l\right\Vert_{\R^N}^2:s_l\in\R^N, \alpha_l\in \R^{d+1}\mbox{, and }\mathbf{T}^Ts_l=0 \right\}$$ for $l=1,2,\cdots,D$, where
\begin{itemize}
    \item $\mathbf{T}$ is an $N\times(d+1)$ matrix whose $(i,j)^{th}$ element is $p_j(\pi_{f_{(n)}}(\mu_{i,N}))$;
    \item $\mu_l=(\mu_{1,N, l},\mu_{2,N, l},\cdots,\mu_{N,N, l})^T$, $\alpha_l=(\alpha_{1, l},\alpha_{2, l},\cdots,\alpha_{d+1, l})^T$, $s_l=(s_{1, l},s_{2, l},\cdots,s_{N, l})^T$ for $l=1,2,\cdots,D$;
    \item $\mathbf{E}$ is an $N\times N$ matrix whose $(i,j)^{th}$ element is $\eta_{4-d}\left( \pi_{f_{(n)}}(\mu_{i,N})-\pi_{f_{(n)}}(\mu_{j,N})\right)$;
    \item $\mathbf{W}=diag(\theta_{1,N},\theta_{2,N},\cdots,\theta_{N,N})$. 
\end{itemize}
Using the Lagrange multiplier method, we can obtain these minimizers by solving the following linear equations. 
\begin{align}\label{smoothing matrix}
\begin{pmatrix}
2\mathbf{EWE}+2\lambda\mathbf{E} & 2\mathbf{EWT} & \mathbf{T}\\
2\mathbf{T}^T\mathbf{WE} & 2\mathbf{T}^T\mathbf{WT} & \mathbf{0}\\
\mathbf{T}^T & \mathbf{0} & \mathbf{0}\\
\end{pmatrix}
\begin{pmatrix}
s_l\\
\alpha_l\\
m_l\\
\end{pmatrix}
=
\begin{pmatrix}
2\mathbf{EW}\mu_l\\
2\mathbf{T}^T\mathbf{W}\mu_l\\
\mathbf{0}\\
\end{pmatrix}  
,\ \ l=1,2,\cdots,D,
\end{align}
where $m_l$ are Lagrange multipliers. The coefficient matrix in (\ref{smoothing matrix}) is symmetric, has many zero elements, and of order $N+2d+2$. Since $N$ is moderate in most applications (see Figure \ref{fig:skeleton} (c)), solving (\ref{smoothing matrix}) is not computationally expensive.  

As detailed in (\ref{wowowow}), we use $\widehat{Q}_{N}$ to estimate $\widehat{f}_\lambda$ for each $\lambda>0$. This procedure shrinks the collection of candidate functions from $C_\infty\bigcap\nabla^{-\otimes2}L^2$ to the one-parameter family $\{\widehat{f}_\lambda\}_{\lambda>0}$. Theorem \ref{bending energy upper bound} shows that this approach prevents curvature singularity at $\lambda=0$. Then we choose an optimal element $\widehat{f}_{\lambda^*}$ in $\{\widehat{f}_\lambda\}_{\lambda>0}$ by using the observed data $\{x_i\}_{i=1}^I$ to tune the family $\{\widehat{f}_\lambda\}_{\lambda>0}$. Specifically, we choose  $\widehat{f}_{\lambda^*}$ which minimizes the MSD $\mathcal{D}(\widehat{f}_\lambda)$ associated with $\{x_i\}_{i=1}^I$, i.e.,
\begin{align}\label{eq: optimal tuning parameter}
\lambda^*=\arg\min_{\lambda>0}\left\{ \mathcal{D}(\widehat{f}_\lambda) \right\}, \mbox{ where }\mathcal{D}(\widehat{f}_\lambda)= \frac{1}{I}\sum_{i=1}^I \left\Vert x_i - \widehat{f}_\lambda\left(\pi_{\widehat{f}_\lambda}(x_i)\right)\right\Vert_{\mathbb{R}^D}^2.
\end{align}
In applications, higher values of the tuning parameter $\lambda$ reduce the effect of corrupting noise. The reduction from $\{x_i\}_{i=1}^I$ to $\widehat{Q}_N$ reduces the corrupting noise and, hence, the corresponding $\lambda$ is expected to be small. Therefore, the estimated optimal $\lambda^*$ tends to be small. Figure \ref{tuning parameter choice} illustrates the relationships between $\log\lambda$, $\log\lambda^*$, and MSD $\mathcal{D}(\widehat{f}_\lambda)$.
\begin{figure}[ht]
	\begin{center}
		\includegraphics[scale=0.305]{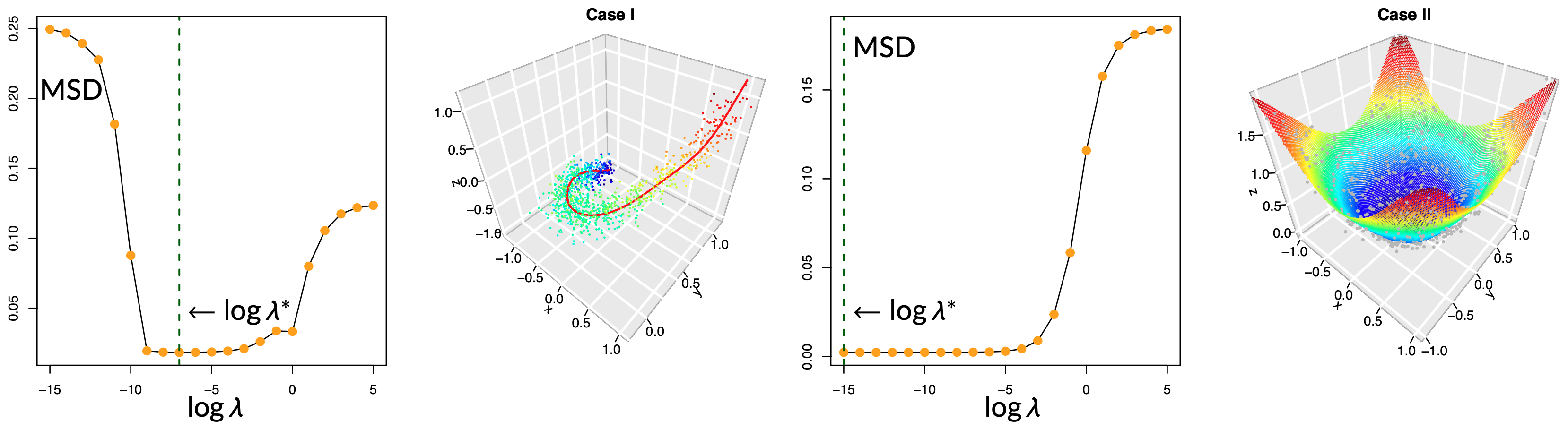}
	\end{center}
	\caption{Points (colored) in Case I are $1000$ realizations of $X$ with $(X\vert T=t)\sim N_3(t,0.1I_{3\times 3})$, $\tau\sim Unif(-1,1)$, $T=(\tau, \tau^2, \tau^3)^T$. The points (gray) in Case II are $1000$ realizations of $X$ with $(X\vert T=t)\sim N_3(t, 0.05I_{3\times 3})$, $\tau_1, \tau_2\sim_{iid}$ $Unif(-1,1)$, $T=(\tau_1, \tau_2, \tau_1^2+\tau_2^2)^T$. Using Algorithm \ref{PME Algorithm}, we fit the points in Case I and Case II, respectively. With tuning parameters $\lambda=e^{k}, k=-15, -14,\cdots, 5$, we plot MSD versus $\log\lambda$ as above. The (green) dash lines indicate optimal tuning parameters. For Case II, the reason why the smallest $\lambda$ is chosen is that the corresponding reduced points $\mu_{j,N}$ (with associated weights $\theta_{j,N}$) have almost no information of the $3$-dimensional corrupting noise $N(t, 0.05 I_{3\times 3})$.}\label{tuning parameter choice}
\end{figure}
\begin{algorithm}[ht]
	\caption{PME Algorithm:}\label{PME Algorithm}
	\begin{algorithmic}[1]
		\INPUT
		(i) Data points $\{x_i\}_{i=1}^I$ in $\R^D$; (ii) a positive integer $N_0<I-1$; (iii) $\alpha, \epsilon, \epsilon^*\in(0,1)$; (iv) candidate tuning parameters $\{\lambda_g\}_{g=1}^G$; (iv) $itr\ge 1$, which is the maximum number of iterations allowed.
		\OUTPUT (i) Analytic formula of $f^*: \R^d\rightarrow\R^D$ determining the fitted manifold $M_{f^*}^d$; (ii) optimal tuning parameter $\lambda^*$.
		\STATE Apply HDMDE (Algorithm \ref{HDMDE algorithm}) with input $(\{x_i\}_{i=1}^I,N_0,\epsilon,\alpha)$ and obtain $N$, $\{\mu_{j,N}\}_{j=1}^N$, and $\{\theta_{j,N}\}_{j=1}^N$.
		\STATE \textbf{Parameterization}: Apply ISOMAP to parameterize the reduced collection $\{\mu_{j,N}\}_{j=1}^N$ by the $d$-dimensional parameters $\{t_j\}_{j=1}^N$. Formally set $\pi_{f_{(0)}}\left(\mu_{j,N}\right)$ $\leftarrow$ $t_j$ for $j=1,2,\cdots,N$.
		\FORALL{$g=1,2,\cdots,G$} 	
		\STATE $\lambda\leftarrow\lambda_g$ and obtain $f_{(1)}$ by solving (\ref{smoothing matrix}).
		\STATE $\mathcal{E}$ $\leftarrow$ $2\times\epsilon^*$, $n\leftarrow1$, and $\mathcal{D}_{\widehat{Q}_N}(f_{(1)})\leftarrow\sum_{j=1}^N \theta_{j,N}\Vert \mu_{j,N} - f_{(1)}(\pi_{f_{(1)}}(\mu_{j,N}))\Vert_{\R^D}^2$.
		\WHILE {$\mathcal{E}\ge\epsilon^*$ and $n< itr$, }
		\STATE Compute $f_{(n+1)}$ from $f_{(n)}$ by solving (\ref{smoothing matrix}) and $$\mathcal{D}_{\widehat{Q}_N}(f_{(n+1)})\leftarrow\sum_{j=1}^N \theta_{j,N} \left\Vert \mu_{j,N} - f_{(n+1)}\left(\pi_{f_{(n+1)}}(\mu_{j,N})\right)\right\Vert_{\R^D}^2.$$
		\STATE $\mathcal{E}\leftarrow\vert[\mathcal{D}_{\widehat{Q}_N}(f_{(n+1)})-\mathcal{D}_{\widehat{Q}_N}(f_{(n)})]/\mathcal{D}_{\widehat{Q}_N}(f_{(n)})\vert$ and $n\leftarrow n+1$.
		\ENDWHILE
		\STATE $\widehat{f}_{g}\leftarrow$ $f_{(n)}$.
		\ENDFOR
                 \STATE $\kappa\leftarrow\max \{\Vert\pi_{\widehat{f}_{g^*}}(x_i)\Vert_{\R^d}:i=1,2,\cdots,I \}$, where $$g^*=\arg\min_{g}\left\{\frac{1}{I}\sum_{i=1}^I \left\Vert x_i - \widehat{f}_g\left(\pi_{\widehat{f}_g}(x_i)\right)\right\Vert^2_{\R^D}: g=1,2,\cdots, G \right\}.$$ 
                 \STATE $f^*(t)=\widehat{f}_{g^*}(\kappa t)$, where the analytic formula of $f^*$ is from (\ref{analytic expression of minimizer}), and $\lambda^*\leftarrow\lambda_{g^*}$.
	\end{algorithmic}
\end{algorithm}

Determining the pair $(N,\lambda^*)$ by HDMDE and (\ref{eq: optimal tuning parameter}) completes our model complexity selection procedure and, hence, the PME algorithm. The PME algorithm presented in Algorithm \ref{PME Algorithm} encapsulates the procedures presented in Sections \ref{section: The Reduction Step of PME} and \ref{explicit formula}. The \texttt{R} code for performing estimation using Algorithm 2 is available at \url{https://github.com/KMengBrown/Principal-Manifold-Estimation.git}. While a rigorous proof of the convergence of Algorithm \ref{PME Algorithm} is outside of the scope of this paper, the algorithm converged in almost all of the simulation studies conducted.


\section{Simulations}\label{simulations}

We compare the PME algorithm to existing methods for simulated data in the following three scenarios with dimension pairs $\left(d=1,D=2\right)$, $\left(d=1,D=3\right)$, and $\left(d=2,D=3\right)$. Simulation analyses in this section are implemented in the \texttt{R} software (\cite{citationR}). In the implementation of PME (Algorithm \ref{PME Algorithm}), we set its inputs as follows: $\lambda_g=\exp(g)$ for $g=-15,-14,\cdots, 5$, $N_0=20\times D$, $\alpha=0.05$, $\epsilon=0.001$, $\epsilon^*=0$, $itr$ values are given in Tables \ref{table 1d2D} and \ref{table 2d3D}. For the first two dimension pairs, we compare PME to two methods: (i) The HS principal curve algorithm using the \texttt{R} function \texttt{principal\_curve} in package \texttt{princurve} (version 2.1.4). Three \texttt{smoother} options - \texttt{smooth$\_$spline}, \texttt{lowess}, and \texttt{periodic$\_$lowess} - are provided in this \texttt{R} function. In each simulation, we apply all the three smoothers and apply the one producing the smallest MSD $\mathcal{D}(f)=\frac{1}{I}\sum_{i=1}^{I}\left\Vert x_i-f\left(\pi_f(x_i)\right)\right\Vert_{\R^D}^2$. (ii) ISOMAP: Parameterize the $I$ data points $\{x_i\}_{i=1}^I$, by $\{t_i\}_{i=1}^I$ using ISOMAP and then fit a map $f^*_{isomap}=\arg\min_{f}\{\frac{1}{I}\sum_{i=1}^{I}\Vert x_i-f\left(t_i\right)\Vert_{\R^D}^2+\lambda^*\Vert\nabla^{\otimes 2} f\Vert_{L^2\left(\R^1\right)}^2: f\in\nabla^{-\otimes 2}L^2(\R^1\rightarrow\R^D)\}$, where $\lambda^*$ is chosen to be the optimal tuning parameter obtained in the corresponding PME fit to make the comparison fair. The minimizer $f^*_{isomap}$ is achieved by cubic smoothing splines (\cite{duchon1977splines}). For $(d=2,D=3)$, we compare PME to two methods: (i) the principal surface (PS) algorithm introduced by \cite{yue2016parameterization}, where the optimal number of basis functions in PS is obtained by the new cross-validation method proposed by \cite{yue2016parameterization}, using the \texttt{R} function for PS provided by the first author of \cite{yue2016parameterization}; (ii) ISOMAP: Conducted as described above for its counterpart with $d=1$, except we apply thin plate splines to achieve the corresponding minimizer (\cite{duchon1977splines}). For all scenarios, the performance measurement of a fitted $f$ is the MSD $\mathcal{D}(f)$. Although we apply ISOMAP to the parameterization step of PME, the implementation of PME has an important advantage compared to solely applying ISOMAP. When directly applying ISOMAP to the observed data, we use the full set of $I$ data points for parameterization, while within PME we apply ISOMAP to parameterize the reduced set $\{\mu_{j,N}\}_{j=1}^N$ in the process of iteration. Since ISOMAP is computationally expensive for high-dimensional data, and the sample size $I$ tends to be much larger than $N$ (see Figure \ref{fig:skeleton} (c)), directly applying ISOMAP to the full observed data is much more time consuming than applying the PME approach. As shown by our simulation studies, this gain in reduction of computation time is accompanied by similar performance of PME as compared to ISOMAP.

Visualizations of the estimation results for some curves and surfaces are presented in Figures \ref{simulation graph d1D2} and \ref{simulation graph d1D3}. When $(d=1, D=2)$, we generate data under four settings in Figure \ref{simulation graph d1D2} and apply all corresponding methods for each of the four cases; when $(d=1, D=3)$ or $(d=2, D=3)$, we generate data using the three different data cloud scenarios in Figure \ref{simulation graph d1D3} and apply all corresponding methods in each of all the three cases. For each method in each case, we run $100$ simulations with simulated data sets of size $I=1000$ and summarize the simulation results in Tables \ref{table 1d2D} and \ref{table 2d3D}, where the mean and standard deviation (sd) of the $100$ simulated MSD in each case for each method are presented. The column ``itr" in these tables shows the number of iterations conducted for each algorithm. Column groups (a) (b) (c) (d) in Table \ref{table 1d2D} correspond to the panels (a) (b) (c) (d) in Figure \ref{simulation graph d1D2}, respectively. Column groups (a) (b) in Table \ref{table 2d3D} correspond to the panels (a) (b) in Figure \ref{simulation graph d1D3}, respectively. The column group (c) in Table \ref{table 2d3D} corresponds to the row (c), i.e., the lower panels, of Figure \ref{simulation graph d1D3}. Except for ISOMAP, all methods take less than ten minutes to run in each simulation in all cases on a PC with a \texttt{2.6 GHz Intel Core i5} processor and \texttt{8 GB 1600 MHz DDR3} memory. In all our simulations, when $d=1$, PME and HS take a similar amount of time to run; when $d=2$, PME and PS take a similar amount of time to run. Further optimization of authors' \texttt{R} code should make the proposed PME algorithm more efficient.  
\begin{table}[ht]
	\centering
	\caption{MSD comparison: $d=1$ and $D=2$. (The unit of mean and sd is $10^{-3}$, and the lowest mean in each column is in bold.)}\label{table 1d2D}
	\begin{tabular}{llllllllllllllll} 
		\hline
		&     &   (a)   &        &  &     &   (b)   &       &  &     &  (c)   &        &  &     &  (d)   &         \\ 
		\cline{2-4}\cline{6-8}\cline{10-12}\cline{14-16}
		Methods & itr & mean  & sd     &  & itr & mean  & sd    &  & itr & mean  & sd     &  & itr & mean  & sd      \\ 
		\hline
		PME     & 20  & 5.995 & 0.4082 &  & 100 & \textbf{40.77} & 1.682 &  & 10  & \textbf{10.09} & 0.6029 &  & 5   & 23.66 & 1.082   \\
		HS      & 300 & 28.33 & 8.690  &  & 200 & 351.8 & 8.702 &  & 100 & 12.96 & 0.4344 &  & 5   & 24.21 & 1.120   \\
		ISOMAP  & 0   & \textbf{5.712} & 0.3297 &  & 0   & 40.97 & 1.802 &  & 0   & 10.12 & 0.4171 &  & 0   & \textbf{23.50} & 0.8739  \\
		\hline
	\end{tabular}
\end{table}
\begin{figure}[ht]
	\begin{center}
		\includegraphics[scale=0.31]{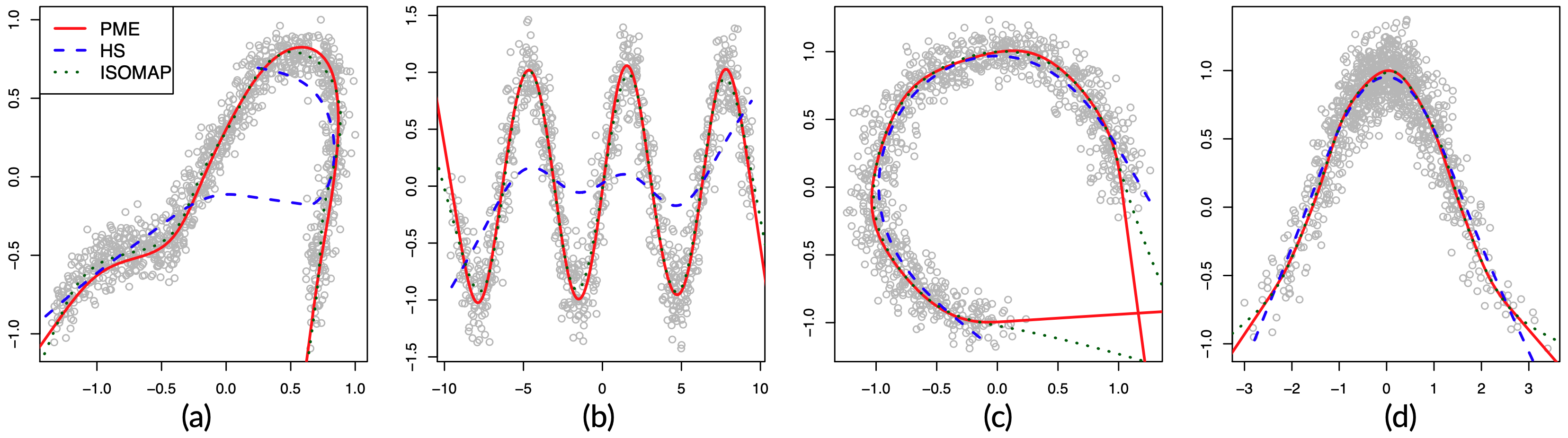}
	\end{center}
	\caption{ Illustration of simulation settings. In each setting, data (in gray) are generated as follows: (a) a $1/4$ part of one slice of a CT data set presented in Section \ref{An Application} is used with added Gaussian noise; (b) realizations of $X$ with $\left(X\vert T=t\right)\sim N_2(t, 0.2 I_{2\times2})$, $T=(\tau, \sin\tau)^T$ and $\tau\sim Unif(-3\pi,3\pi)$; (c) realizations of the $X$ in Figure \ref{fig:skeleton} (b) (without the $10$ outliers); (d) realizations of $X$ with $\left(X\vert T=t\right)\sim N_2(t, 0.15I_{2\times 2})$, $T=(\tau,\cos\tau)^T$ and $\tau\sim N(0,1)$.}\label{simulation graph d1D2}
\end{figure}
\begin{table}[ht]
	\centering
	\caption{MSD comparison: $d=1, 2$ and $D=3$. (The unit of mean and sd is $10^{-3}$, and the lowest mean in each column is in bold.)}\label{table 2d3D}
	\begin{tabular}{lllllllllllllllllllllllll} 
		\cline{1-9}\cline{12-16}
		$d=1$    &  &     & (a)  &    &  &     & (b)  &    &  &  &     $d=2$    &  &     &   (c)   &   \\ 
		\cline{1-1}\cline{3-5}\cline{7-9}\cline{12-12}\cline{14-16}
		Methods &  & itr & mean & sd &  & itr & mean & sd &  &  & Methods &  & itr & mean & sd\\ 
		\cline{1-9}\cline{12-16}
		PME     &  &  100   &   \textbf{18.58}   &  0.6023  &  &  100   &   5.320   &  0.2115  &  &  & PME     &  &  10   &   2.522   & 0.1138\\
		HS      &  &  200   &   21.23   &  0.6294  &  &  500   &   88.03   &  0.7479  &  &  & PS      &  &   10  &   2.520   &  0.1137\\
		ISOMAP  &  &   0  &   19.52   &  0.6163  &  &  0   &   \textbf{5.214}   &  0.1661  &  &  & ISOMAP  &  &   0  &   \textbf{2.496}   &  0.1103\\
		\cline{1-9}\cline{12-16}
	\end{tabular}
\end{table}
\begin{figure}[ht]
	\begin{center}
		\includegraphics[scale=0.33]{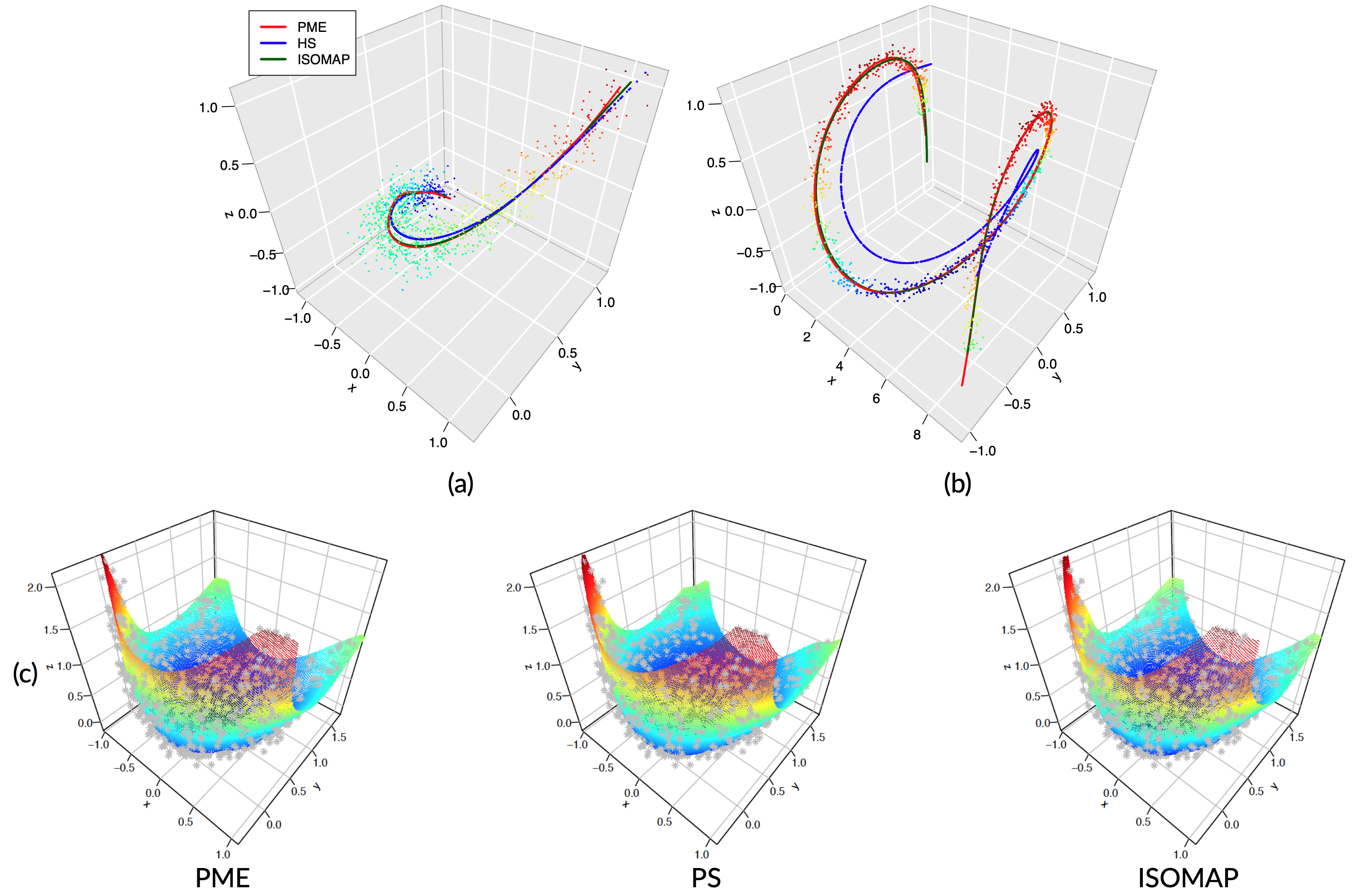}
	\end{center}
	\caption{In each case, $X$ are generated as follows: (a) using the same approach as Case I in Figure \ref{tuning parameter choice}; (b) $(X\vert T=t)\sim N_3(t, 0.05I_{3\times3})$, $T=(\tau, \cos\tau, \sin\tau)^T$, $\tau\sim Unif(\pi/2, 6\pi)$; the three lower panels share the same data (gray), where $(X\vert T=t)\sim N_3(t, 0.05I_{3\times3})$, $\tau_1,\tau_2\sim_{iid} Unif(-1, 1)$, and $T=(\tau_1, \frac{1}{2}(\tau_2+\sqrt{3}(\tau_1^2+\tau_2^2)), \frac{1}{2}(\tau_1^2+\tau_2^2-\sqrt{3}))^T$.} \label{simulation graph d1D3}
\end{figure}

\noindent\textbf{Simulation results}: (i) For ($d=1, D=2$), Figure \ref{simulation graph d1D2} (a) (b) show that PME performs much better than HS. Figure \ref{simulation graph d1D2} (c) (d) show that PME performs slightly better than HS. ISOMAP and PME perform similarly well in all four cases. The difference between them is visible only near the tails of data clouds. (ii) For ($d=1,D=3$), Figure \ref{simulation graph d1D3} (a) shows that the three methods perform similarly well. Figure \ref{simulation graph d1D3} (b) shows that PME and ISOMAP perform similarly well, and both of them perform much better than HS. (iii) For ($d=2,D=3$), Figure \ref{simulation graph d1D3} (c) shows that PME, PS, and ISOMAP perform equally well. Tables \ref{table 1d2D} and \ref{table 2d3D} support these conclusions. In conclusion, PME performs either significantly or marginally better than HS, and PME is not inferior to ISOMAP. However, ISOMAP is extremely time consuming in all scenarios compared to other methods. If we increase the size of simulated data sets, applying ISOMAP becomes infeasible. 


\section{Interior Identification}\label{Interior classifier}

In this Section, we propose an algorithm to identify the interiors of circle-like curves ($d=1, D=2$) and cylinder/ball-like surfaces ($d=2, D=3$). Examples of such curves and surfaces are presented in  Figure \ref{Classifier illustration}. In many applications, the target is not the surface of an object, but its interior. For example, radiation therapists may be interested in identifying the interior of a tumor, which contains malignant cells. We propose an interior identification method based on PME.

Let $\underline{M}^d$ denote a circle-like curve ($d=1$) or cylinder/ball-like surface ($d=2$) contained in a domain $\mathcal{E}\subset\R^D$, e.g., the punched sphere in Figure \ref{Classifier illustration} (b) is contained in a cube. The main idea of the interior identification approach is as follows: \textbf{Stage 1}, we decompose $\mathcal{E}$ into several potentially overlapping subsets, i.e., $\mathcal{E}=\bigcup_{s=1}^S E[s]$, where $E[s]$ are subsets of $\mathcal{E}$; \textbf{Stage 2}, we identify the interior of each piece $\underline{M}^d\bigcap E[s]$; \textbf{Stage 3}, we ``glue" the piecewise interior identification results of all $\underline{M}^d\bigcap E[s]$ using the \textit{10-nearest neighborhood classifier} and obtain the interior estimation of $\underline{M}^d=\bigcup_{s=1}^S(\underline{M}^d\bigcap E[s])$. In Stage 2, we assume that, for each index $s$, there exists an $f_{\underline{s}}:\R^d\rightarrow\R^D$ such that $\underline{M}^d\bigcap E[s]=M^d_{f_{\underline{s}}}\bigcap E[s]$, where the manifold $M^d_{f_{\underline{s}}}$ is defined by Definition \ref{Def 2} and estimated using PME.

We first propose the interior identification approach for each piece $M_f^d\bigcap E$, where $f:\R^d\rightarrow\R^D$ and $E$ is a sub-domain of $\mathcal{E}$. Let $\overrightarrow{\pmb{n}}(t)$ denote a normal vector of $M_f^d$ at point $f(t)$. For example, 
\begin{align*}
    & \overrightarrow{\pmb{n}}(t)=\left(-\frac{df_2}{dt}(t), \frac{df_1}{dt}(t)\right)^T \mbox{ when $d=1$ and $D=2$, and} \\
    & \overrightarrow{\pmb{n}}(t)=\left(\frac{\partial f_2}{\partial t_1}\frac{\partial f_3}{\partial t_2}-\frac{\partial f_3}{\partial t_1}\frac{\partial f_2}{\partial t_2},\frac{\partial f_3}{\partial t_1}\frac{\partial f_1}{\partial t_2}-\frac{\partial f_1}{\partial t_1}\frac{\partial f_3}{\partial t_2},\frac{\partial f_1}{\partial t_1}\frac{\partial f_2}{\partial t_2}-\frac{\partial f_2}{\partial t_1}\frac{\partial f_1}{\partial t_2}\right)^T \mbox{ when $d=2$ and $D=3$.}
\end{align*}
Computing $\overrightarrow{\pmb{n}}(t)$ is possible since we have the analytic formula (\ref{analytic expression of minimizer}). For a fixed $\xi\in\R^D$, $Orit(\xi, f)=sgn\left\{\left[ f\left(\pi_{f}(\xi)\right)-\xi\right]^T \overrightarrow{\pmb{n}}\left(\pi_{f}(\xi)\right)\right\}$ is called the \textit{orientation} of $\xi$ with respect to $f$, where $sgn(\cdot)$ is the sign function. Let $c^*$ be a predetermined point indicating the interior side of $M_f^d\bigcap E$. It is called the \textit{reference point}. Then all the points in $E$ sharing the same orientation with $c^*$ are identified as interior points, i.e., the interior part of $M^d_f\bigcap E$ is estimated by $\mathcal{I}\left(f, c^*\right)\bigcap E$, where $\mathcal{I}\left(f, c^*\right)=\left\{\xi\in\R^D:Orit(\xi,f)\times Orit(c^*,f)>0\right\}$. A geometric illustration is presented in Figure \ref{Classifier illustration} (a).
\begin{figure}[ht]
	\begin{center}
		\includegraphics[scale=0.27]{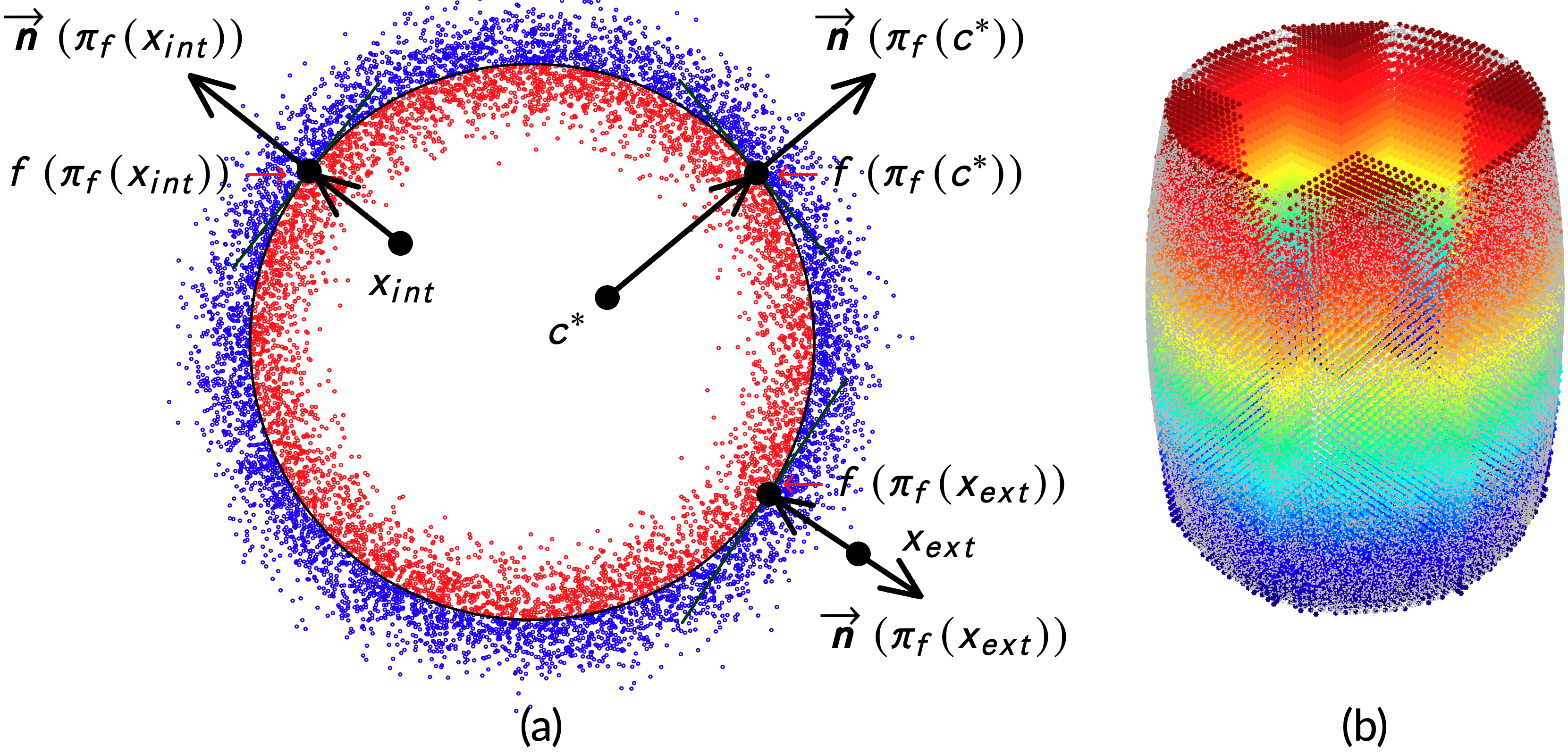}
	\end{center}
	\caption{(a) An illustration of $\pmb{n}(\pi_f(\cdot))$, $f(\pi_f(\cdot))$, and the reference point $c^*$. (b) A simulation example, where $10000$ data points are from $(\sin\tau_1\times\cos\tau_2, \sin\tau_1\times\sin\tau_2,\cos\tau_1)^T$, with $\tau_1\sim Unif(\pi/4,3\pi/4),\tau_2\sim Unif(0,2\pi)$. There is no $3$-dimensional corrupting noise in these data. The colored points indicate the $\xi_j$ identified as interior of the punched sphere. To illustrate the boundaries of the cubes $E[k]$, we omit all interior $\xi_j$ outside of $\mathcal{E}=\bigcup_{k=1}^8 E[k]$.}\label{Classifier illustration}
\end{figure}

Secondly, we explain the interior identification approach for the entire $\underline{M}^d$ by an example - fitting the $I=10000$ data points $\{x_i\}_{i=1}^I$ in Figure \ref{Classifier illustration} (b) (gray points). These data points are simulated from a punched sphere, which is the $\underline{M}^d$ of interest. The reference point $c^*=(0,0,0)^T$ is the centroid of the sphere. The points to be identified are grid-points $\xi_j$, such as the colored points in Figure \ref{Classifier illustration} (b). In this example, we show constructions of the domain $\mathcal{E}$ and its subsets $E[s]$. We identify the grid-points $\xi_j$ interior of this punched sphere using the following procedure.\\
\textbf{Step 1:} For each $3$-dimensional vector $x_i=(x_{i,1}, x_{i,2}, x_{i,3})^T$ where $x_{i,l}$ denotes the $l^{th}$ component of $x_i$, let $(\phi_i, r_i)^T$ be the polar coordinate of the $2$-dimensional vector $(x_{i,1}, x_{i,2})^T$ and $\phi_i$ be the corresponding angle component. Partition $\{x_i\}_{i=1}^I$ into $8$ subsets by $\mathcal{Z}[k]=\{x_i: \frac{(k-1)\pi}{4}\le \phi_i< \frac{k\pi}{4}\}$ for $k=1,2,\cdots, 8$. \\
\textbf{Step 2:} Define the cubes $E[k]=\prod_{l=1}^3 \left[\inf_{x_i\in\mathcal{Z}[k]}x_{i, l}, \sup_{x_i\in\mathcal{Z}[k]}x_{i, l}\right]$ for all $k$. Then $\mathcal{E}=\bigcup_{k=1}^8 E[k]$ contains all $x_i$.\\
\textbf{Step 3:} Fit an $f_1$ to data in $\mathcal{Z}[8]\bigcup\mathcal{Z}[1]$ and an $f_{k}$ to data in $\mathcal{Z}[k-1]\bigcup\mathcal{Z}[k]$ for all $k=2,3,\cdots, 8$ using PME.\\
\textbf{Step 4:} For each $k$, define $x^*_k=\frac{1}{\#\mathcal{Z}[k]}(\sum_{x_i\in\mathcal{Z}[k]}x_i)\in\mathbb{R}^3$, where $\#\mathcal{Z}[k]$ denotes the number of elements in $\mathcal{Z}[k]$. \\
\textbf{Step 5:} All grid-points $\xi_j\notin\mathcal{E}$ are identified as exterior. For each $\xi_j\in\mathcal{E}$, compute $k=\arg\min_{k'=1,2,\cdots,8}\{\Vert \xi_j-x^*_{k'}\Vert_{\R^3}\}$. Since both $f_k$ and $f_{k+1}$ fit data in $\mathcal{Z}[k]$, there are three possible scenarios:\\ 
(i) $\xi_j\in\mathcal{I}\left(f_k, c^*\right)\bigcap\mathcal{I}\left(f_{k+1}, c^*\right)$, i.e., $\xi_j$ is identified as interior by both $f_k$ and $f_{k+1}$, then $\xi_j$ is identified as interior and labeled by ``int";\\ (ii) $\xi_j\notin\mathcal{I}\left(f_k, c^*\right)\bigcup\mathcal{I}\left(f_{k+1}, c^*\right)$, i.e., $\xi_j$ is identified as exterior by both $f_k$ and $f_{k+1}$, then $\xi_j$ is identified as exterior and labeled by ``ext";\\
(iii) $\xi_j$ satisfies neither the previous two scenarios, then we identify $\xi_j$ by applying 10-nearest neighborhood classifier to the labeled training set $\left\{(\xi_q, lab_q):\mbox{$\xi_q\in E[k]$ and $\xi_q$ satisfies scenario (i) or (ii)}\right\},$ where $lab_q\in\{\mbox{``int", ``ext"}\}$ is the label of $\xi_q$.

The performance of this interior identification procedure is shown in Figure \ref{Classifier illustration} (b). Since we know the true punched sphere generating data, the true interior/exterior labels of grid-points $\xi_j \in \mathcal{E}=\bigcup_{k=1}^8 E[k]$ with respect to this punched sphere are known. The identification error rate - the proportion of incorrectly estimated labels - is less than $0.1\%$. In the illustrative example in Figure \ref{Classifier illustration}, we automatically and evenly divide data $x_i$ into eight subcollections. In general, depending on the shape of the observed data, we may need to divide data into more/fewer subcollections. Additionally, an uneven division might be suitable for some data sets. For example, we may conduct a finer division in a region containing a large number of data points than in a region containing only a few data points. Determining the number of subcollections and division precision in individual regions is left for future research. Additionally, future research may extend our proposed methods for identifying the interiors of a more general set of manifolds.

\begin{figure}[ht]
	\begin{center}
		\includegraphics[scale=0.31]{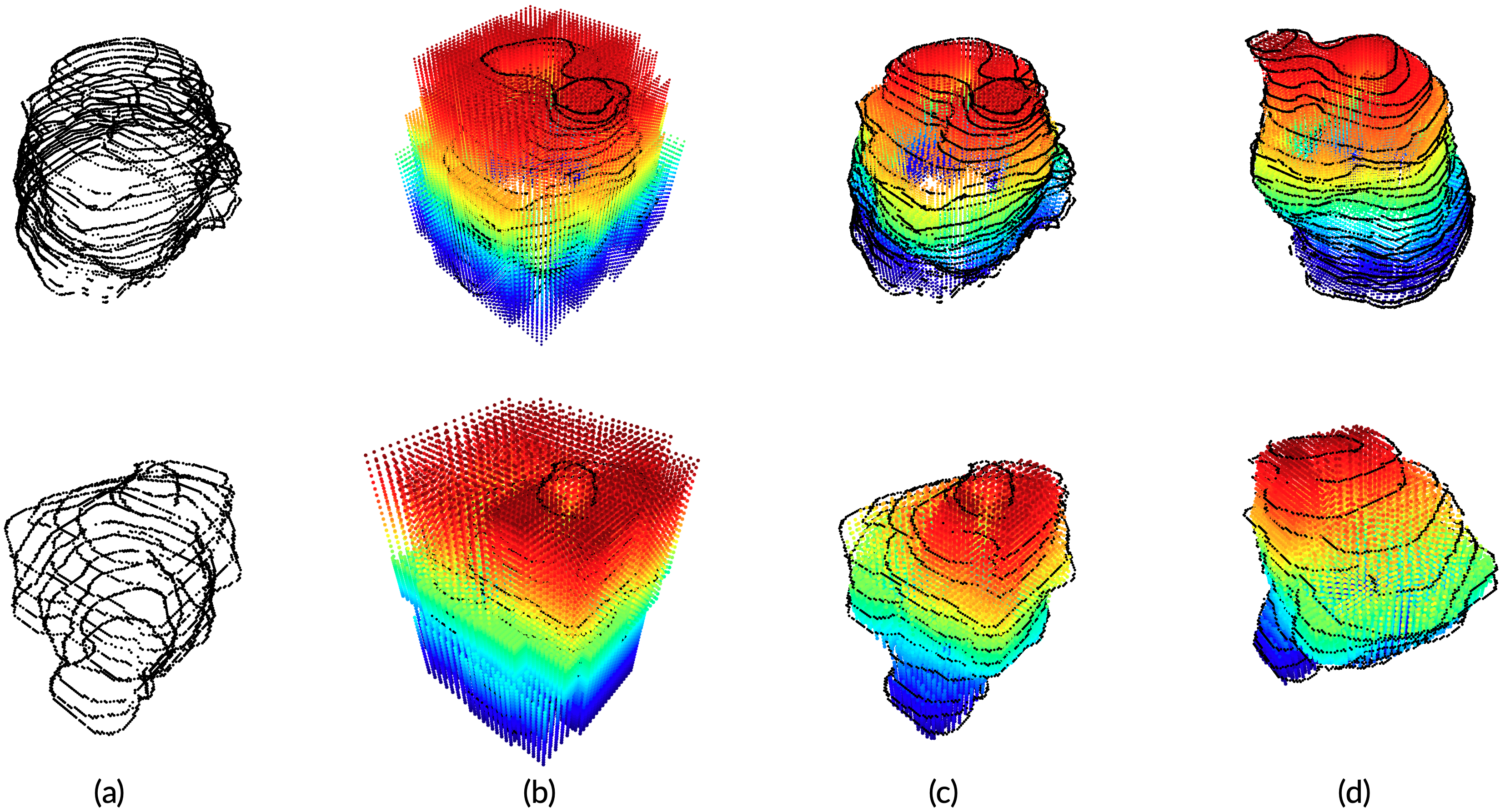}
	\end{center}
	\caption{The black points in column (a) denote the CT data of two tumors. The colored points in column (b) denote the points to be identified. The colored points in column (c,d) denote the points identified as interior of the tumors. The last two columns show different angles of the tumors.} \label{tumor interior classifier}
\end{figure}

\section{Analysis of Lung Cancer Tumor Data}\label{An Application}

In this section, we consider tumor surface estimation using computed tomography (CT) scans collected from patients with lung cancer and the identification of tumor interior in the context of radiation therapy. We analyzed two tumor data sets from a publicly available database collected for 422 patients with non-small cell lung cancer at the MAASTRO Clinic (Maastricht, The Netherlands) and available at \url{http://www.cancerimagingarchive.net/}. Spiral CT scans of the thoracic region with a $3mm$ slice thickness are obtained for each study participant. In addition, the masks of the tumor hand segmented by a radiologist are provided in the database. The hand segmentation result is a collection of voxels (3-dimensional counterparts of pixels) in 3-dimensional space marked by the radiologist as points on the surface of the tumor. The details on imaging parameters are available on the website, the references provided therein, and are not repeated in this section. The vertices of the tumors for the two participants are presented in Figure \ref{tumor interior classifier} (a). Given that we only have a collection of points on the tumor's surface, it is necessary to estimate the surface of the tumor fitted to the manually selected vertices on the surface of the tumor.  In addition to estimating the tumor surface, it may be of interest to identify the interior area of the tumor. For example, in radiation therapy, ionizing radiation is used to control or kill cancer cells. To avoid harming healthy tissue with unnecessary doses of radiation, identifying the interior region of a tumor is important. Since the shape of the tumors is similar to a punched sphere, we apply the same procedures as in the example in Section \ref{Interior classifier} to identify the interior part of these tumors.
\begin{figure}[ht]
	\begin{center}
		\includegraphics[scale=0.27]{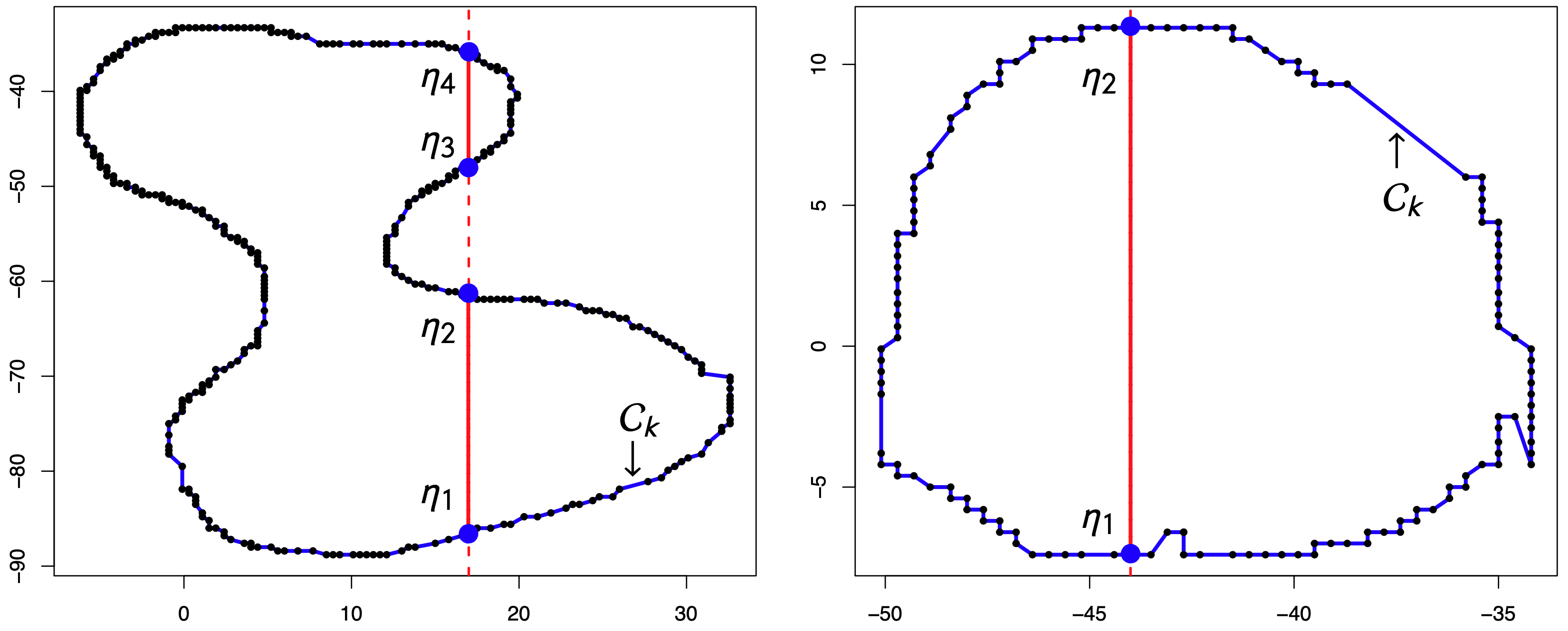}
	\end{center}
	\caption{Left: a single slice from the CT data for one subject (presented in the upper panels of Figure \ref{tumor interior classifier}). Right: a single slice from the CT data for another subject (presented in the lower panels of Figure \ref{tumor interior classifier}).} \label{Figure: 2nd interior method}
\end{figure}

The interior identification result is in Figure \ref{tumor interior classifier}. Visually, the proposed method can properly identify interior points. We also use a simple approach to identify tumor interior points given the surface voxels and obtain a rough estimate of the validity of our proposed method. By its nature, the CT data is a collection of grid points in a 3D box $\mathcal{B}=[X_L, X_U] \times [Y_L, Y_U] \times [Z_L, Z_U]$ along with the intensities of all voxels in this grid. Let $\mathcal{X}^k=\{\xi^k_j\}_{j}$ be the collection of 3D data points in the $k^{th}$ slice of the CT scan (all $k$ here indicate the $k^{th}$ slice). All the points in $\mathcal{X}^k$ share the same $Z$-coordinate $z^k \in [Z_L, Z_U]$. To identify the interior of the tumor in the rectangle $[X_L, X_U] \times [Y_L, Y_U] \times \{z^k\}$, e.g., the rectangles in Figure \ref{Figure: 2nd interior method}, we use a linear interpolation to connect consecutive points $\xi^k_j$ and $\xi^k_{j+1}$. As a result, we obtain a piecewise linear and closed curve $\mathcal{C}_k$. The curve $\mathcal{C}_k$ (the blue curves in Figure \ref{Figure: 2nd interior method}) roughly indicates the boundary of the tumor in this slice. For any $x^k \in [X_L, X_U]$, let $\{\eta_1, \eta_2, \cdots\}$ be the intersection of the line segment $\{x^k\} \times [Y_L, Y_U]$ and $\mathcal{C}_k$ (the blue dots in Figure \ref{Figure: 2nd interior method}). The points on line segments of the form $$\bigcup_{l}\Big\{\lambda\eta_{2l-1} + (1-\lambda)\eta_{2l}: \lambda \in [0,1]\Big\}$$ 
are identified as the tumor's interior. These line segments are the solid red ones in Figure \ref{Figure: 2nd interior method}. 

Finally, for each of the $50\times50\times\#\{z_k\}$ grid points ($\#\{z_k\}$ denotes the number of slices), we compare the labels given by the rough approximation approach in the preceding paragraph and the PME based interior identification method, respectively. For the two tumor data sets in Figure \ref{tumor interior classifier}, $97.1\%$ of the grid points are given the same labels by these two methods for subject 1 (top panel) and $95.4\%$ for subject 2 (lower panel). Hence, we conclude that these two methods perform similarly in these two examples. However, the rough approximation method has major shortcomings. If the number of points identified by the hand segmentation is small, the line segments in the rough approach will result in a poor estimate of the tumor surface, leading to poor interior/exterior classification performance. Additionally, any outlier surface voxels will potentially have significant negative effects on the rough classifier, while the PME approach is robust to the effects of outliers. Even though the rough approach has these limitations, we considered comparing it to our proposed approach as we have no gold standard classifier to illustrate the performance of our algorithm.


\section{Conclusions and Discussion}

In this paper, we propose a framework of principal manifolds for arbitrary intrinsic dimensions using Sobolev spaces. A Sobolev embedding theorem guarantees the regularity of principal manifolds. To reduce the computational cost and the effects of outliers, and to select model complexity, we propose the HDMDE algorithm motivated by \cite{eloyan2011smooth}. Future work may evaluate the performance of HDMDE as a density estimation technique in high-dimensional settings. Based on HDMDE, we develop the PME algorithm to estimate the newly proposed principal manifolds with intrinsic dimensions $d\le3$. We use simulations to compare PME to existing methods for scenarios with dimension pairs $(d=1,D=2)$, $(d=1,D=3)$, and $(d=2,D=3)$. These simulations illustrate that PME performs better than HS in many scenarios in the sense of minimizing MSD, ISOMAP is too computationally expensive compared to PME, and PME is not inferior to PS. However, PS is only defined for $d=2$. Additionally, PS does not provide an explicit and simple formula of the map $f:\R^d\rightarrow\R^D$ defining estimated $M_f^d$, while we obtain such a formula using PME. We apply PME to radiation therapy by identifying the interiors of tumors, which are targets of ionizing radiation. 

HS principal curves, which apply to data in Euclidean spaces, have been generalized to data on \textit{Riemannian manifolds} (e.g., \cite{hauberg2015principal} and \cite{kim2020spherical}). In the meantime, PCA has been generalized to geodesics in \textit{Wasserstein spaces} (e.g., \cite{boissard2015distribution} and \cite{seguy2015principal}). Future work may investigate the generalization of our proposed framework, for which the HS principal curve algorithm and PCA are special cases, to counterparts on Riemannian manifolds and in Wasserstein spaces. \cite{kirov2017multiple} introduced a \textit{multiple penalized principal curve} framework permitting a fitted $1$-dimensional structure to consist of several disconnected curves. Our proposed framework may be generalized to fit a high-dimensional structure consisting of several disconnected components.

\section{Acknowledgements}

The authors thank the editorial team and the anonymous referees for their insightful comments that significantly improved the quality of this paper. Additionally, the authors thank Dr. Matthew T. Harrison for his comments that improved this paper's readability and Dr. Chen Yue for providing the \texttt{R} function for implementing the PS algorithm. The project was supported by Grant Number 5P20GM103645 from the National Institute of General Medical Sciences.

\newpage

\section{Appendix}\label{proofs appendix}

\textbf{Proof of Theorem \ref{fundation of projection index}}: Since $\Vert f(t) \Vert_{\R^D}\rightarrow\infty$ as $\Vert t\Vert_{\R^d}\rightarrow\infty$, there exists $M>0$ such that $\Vert x-f(t)\Vert_{\R^D}>1+dist(x,f)$ if $\Vert t\Vert_{\R^d}>M$. Then $dist(x,f)=\inf\{\Vert x-f(t)\Vert_{\R^D}: t\in \overline{B_d}(0,M)\}$, where $\overline{B_d}(0,M)=\{t\in\R^d:\Vert t\Vert_{\R^d}\le M\}$. The compactness of $\overline{B_d}(0,M)$ implies that there exists $t^*\in\overline{B_d}(0,M)$ so that $dist(x,f)=\Vert x-f(t^*)\Vert_{\R^D}$, then $\mathcal{A}_f(x)$ is nonempty. We have
\begin{align*}
    \mathcal{A}_f(x)=\left\{t\in\overline{B_d}\left(0,M\right):\left\Vert x-f(t)\right\Vert_{\R^D}\le dist(x,f)\right\}=\overline{B_d}(0,M)\bigcap f^{-1}\left(\overline{B_D}\left(x,dist(x,f)\right)\right),
\end{align*}
where $\overline{B_D}(x, dist(x,f))=\{x'\in\R^D:\Vert x-x'\Vert_{\R^D}\le dist(x,f)\}$ is compact, and $f^{-1}\left(\overline{B_D}(x,dist(x,f))\right)$ is closed as $f$ is continuous. Since $\overline{B_d}(0,M)$ is bounded and closed, $\mathcal{A}_f(x)$ is compact. \qed

\noindent\textbf{Proof of Theorem \ref{measurability of projection indices}} (i) can be proved following the method used to prove Theorem 4.1 in \cite{hastie1984principal}. (ii) $\phi(x)=dist(x,f)$, then $\phi\in C(B)$ and $m^*=\max_{x\in B}\phi(x)<\infty$. Let $\zeta(t)=\inf_{x\in B}\Vert x - f(t)\Vert_{\R^D}$, then $\zeta(t)\ge \Vert f(t)\Vert_{\R^D}-\sup_{x\in B}\Vert x\Vert_{\R^D}\rightarrow\infty$ as $\Vert t\Vert_{\R^d}\rightarrow\infty$. $\exists M>0$ so that $\zeta(t)>m^*$ if $\Vert t\Vert_{\R^d}\ge M$. For all $x\in B$, $dist(x,f)\le m^*<\zeta(t)\le\Vert x - f(t)\Vert_{\R^D}$ if $\Vert t\Vert_{\R^d}\ge M$, which implies $\{t:\Vert t\Vert_{\R^d}\ge M\}\cap[\bigcup_{x\in B}\mathcal{A}_f(x)]=\emptyset$. Then $\{\pi_f(x):x\in B\} \subset \{t:\Vert t\Vert_{\R^d}<M\}$. (iii) Let $\Pi=f\circ\pi_f:\mathbb{R}^D\rightarrow M^d_f$. Since $f$ has no ambiguity point on $U$, Theorem 1.3 of \cite{dudek1994nonlinear} implies that $\Pi$ is continuous on $U$. Since $f^{-1}$ is continuous, $\pi_f=f^{-1}\circ\Pi$ is continuous on $U$. \qed

\noindent\textbf{Proof of Theorem \ref{principal manifolds and PCA}}: That $\mathcal{K}_{\infty,\mathbb{P}}(f)<\infty$ only if $\left\Vert\nabla^{\otimes 2}f\right\Vert_{L^2(\R^d)}=0$ implies the generalized (not classical) derivatives $\frac{\partial^2 f_l}{\partial t_i \partial t_j}=0$, for $1\le i,j\le d$ and $l=1,2,\cdots,D$, almost everywhere. From Lemma \ref{space equivalence} and the Corollary 3.32 of \cite{adams2003sobolev}, $f$ equals an affine function almost everywhere. The continuity of $f$ implies that $f$ equals this affine function exactly everywhere. Then $f\left(\pi_f(X)\right)$ is the projection of $X$ to some hyperplane. Therefore, $\inf_{f\in\mathscr{F}(\mathbb{P})}\mathcal{K}_{\infty,\mathbb{P}}(f)=\inf_{\pmb{C}\in\mathscr{P},a\in\mathbb{R}^D}\mathbb{E}\Vert X-(\pmb{I}-\pmb{C})a-\pmb{C}X\Vert_{\mathbb{R}^D}^2$, where $\mathscr{P}=\{\pmb{C}\in\mathbb{R}^{D\times D}: \pmb{C}^2=\pmb{C}^T=\pmb{C}, rank(\pmb{C})=d\}$ is the collection of projection matrices of rank $d$. Then $\inf_{f\in \mathscr{F}(\mathbb{P})}\mathcal{K}_{\infty,\mathbb{P}}(f)$ is equal to $\inf\{\mathbb{E}\left\Vert (\pmb{I}-\pmb{C})(X-a)\right\Vert_{\mathbb{R}^D}^2: \pmb{C}\in\mathscr{P}, a\in\mathbb{R}^D\}=\inf\{ tr [(\pmb{I}-\pmb{C})\pmb{U}\pmb{D}\pmb{U}^T(\pmb{I}-\pmb{C})^T]+\left\Vert(\pmb{I}-\pmb{C})\left(\mathbb{E}X-a\right)\right\Vert^2_{\R^D}: a\in\mathbb{R}^D, \pmb{C}\in\mathscr{P}\},$ where $\pmb{U}=(\pmb{v}_1, \pmb{v}_2, \cdots, \pmb{v}_D)$ and $\pmb{D}=diag(e_1, e_2, \cdots, e_D)$. The minimum is achieved by the minimizer $(\pmb{C}^*, a^*)$, where $\pmb{C}^*$ is the projection matrix to the subspace $\{\sum_{i=1}^d \alpha_i\pmb{v}_i:\alpha_i\in\mathbb{R}^1\}$ and $a^*$ satisfies $(\pmb{I}-\pmb{C}^*)(\mathbb{E}X-a^*)=0$. Then the minimizer hyperplane is $\{(\pmb{I}-\pmb{C}^*)a^*+\pmb{C}^*x:x\in\mathbb{R}^D\}=\{\mathbb{E}X+\sum_{i=1}^d\alpha_i\mathbf{v}_i: \alpha_i\in\R^1\}$.
 \qed

\noindent \textbf{Derivation of (\ref{EM iteration 3})}: Let $\{Z_i\}_{i=1}^I$ define independent latent random variables taking values in $\{1,2,\cdots, N\}$, such that  $\left(X_i\vert\theta_N, Z_i=z_i\right)\sim\psi_{\widehat{\sigma}_N}\left(x_i-\mu_{z_i, N}\right)dx_i$ and $\left(Z_i\vert\theta_N\right)\sim\theta_{z_i, N}(\sum_{j=1}^N\delta_{j}(dz_i))$ for $i=1,2,\cdots,I$. In other words, the latent variable $Z_i$ indicates the class membership of the $i^{th}$ observation in the mixture. Then we have $(X_i,Z_i)\vert\theta_N\sim\theta_{z_i, N}\psi_{\widehat{\sigma}_N}(x_i-\mu_{z_i, N})[dx_i\times\sum_{j=1}^N\delta_{j}(dz_i)]$ and 
$\mathbb{P}\left(Z_i=j\vert \theta_N,X_i=x_i\right)=w_{ij}(\theta_N).$ 
The complete likelihood of $\left\{(X_i,Z_i)\right\}_{i=1}^I$ with respect to the product measure $\prod_{i=1}^I\{dx_i\times\sum_{j=1}^N\delta_{j}(dz_i)\}$ is $L_C(\mathbf{\theta}_N\vert x,z)=\prod_{i=1}^I\theta_{z_i, N}\times\psi_{\widehat{\sigma}_N}\left(x_i-\mu_{z_i, N}\right)$. For a fixed $\theta_N^{(k)}\in\Theta_N$, in the E-step of EM algorithm, we construct
\begin{align*}
    Q\left(\theta_N\vert\theta_N^{(k)}\right) & =\mathbb{E}\left(\log L_C\left(\theta_N\vert X,Z\right)\big\vert X=x,\theta_N^{(k)}\right)\\
    & =\sum_{i=1}^I\sum_{j=1}^N \left\{w_{ij}\left(\theta_N^{(k)}\right)\log\left(\psi_{\widehat{\sigma}_N}\left(x_i-\mu_{z_i, N}\right)\right)+w_{ij}\left(\theta_N^{(k)}\right)\log\theta_{j,N}\right\}.
\end{align*}
We implement the constraints $\int_{\R^D} x p(x\vert\theta_N)dx = \overline{x}$ and $\sum_{j=1}^{N}\theta_{j,N}=1$ and obtain the Lagrangian
$$Q_\rho\left(\theta_N\vert\theta_N^{(k)}\right)=Q\left(\theta_N\vert\theta_N^{(k)}\right)+\rho_1\left(1-\sum_{j=1}^N\theta_{j,N}\right)+\rho_2^T\left(\overline{x}-\sum_{j=1}^N\theta_{j,N}\mu_{j,N}\right)$$ 
for $\rho_1\in\mathbb{R}^1, \rho_2\in\mathbb{R}^D$. Taking derivatives of $Q_\rho(\theta_N\vert\theta_N^{(k)})$, we obtain
\begin{align*}
    & \frac{\partial Q_\rho}{\partial\theta_{j,N}}=\frac{1}{\theta_{j,N}}\sum_{i=1}^I w_{ij}\left(\theta_N^{(k)}\right)-\rho_1-\rho_2^T\mu_{j,N}=0 \mbox{ for all $j$ and }\\
    & \frac{\partial Q_\rho}{\partial \rho_1}=1-\sum_{j=1}^N\theta_{j,N}=0,\ \ \ \ \ \ \ \ \frac{\partial Q_\rho}{\partial \rho_2}=\overline{x}-\sum_{j=1}^N\theta_{j,N}\mu_{j,N}=0.
\end{align*}
The resulting equations for estimating $\theta_{j,N}$ are
\begin{align*}
    \theta_{j,N}=\frac{\sum_{i=1}^I w_{ij}\left(\theta_N^{(k)}\right)}{\rho_1+\rho_2^T\mu_{j,N}},\ \ \ \ \ \sum_{j=1}^N \left(\frac{\sum_{i=1}^Iw_{ij}\left(\theta_N^{(k)}\right)}{\rho_1+\rho_2^T\mu_{j,N}} \right)=1,\ \ \ \ \ \sum_{j=1}^N \left(\frac{\sum_{i=1}^Iw_{ij}\left(\theta_N^{(k)}\right)}{\rho_1+\rho_2^T\mu_{j,N}}\right)\mu_{j,N}=\overline{x}
\end{align*}
for all $j$, whose solution is given by (\ref{EM iteration 3}). \qed

\noindent\textbf{Proof of Theorem \ref{density estimation}}: Define the weights $\theta_{j,N}=\int_{A_{j,N}}p(\mu)d\mu$, then $\int_{\R^D} p(\mu)d\mu=1$ implies that $\theta_N$ is in the probability simplex $\Theta_N$. By \textit{Minkowski's inequality} (Theorem 2.9 of \cite{adams2003sobolev}), we have
\begin{align*}
    \left\Vert p_N\left(\cdot\vert\theta_N\right)-p\right\Vert_{L^q(\R^D)} & \le \left(\sum_{j=1}^N \int_{A_{j,N}} \left\Vert\psi_{\sigma_N}\left(\cdot-(\mu_{j,N}-\mu)\right)-\psi_{\sigma_N} \right\Vert_{L^q(\R^D)} p(\mu)d\mu \right) +  \left\Vert\psi_{\sigma_N}*p-p \right\Vert_{L^q(\R^D)}\\
    & =:I_N+II_N.
\end{align*}
Since $\mu$ and $\mu_{j,N}$ are in $A_{j,N}$ and $diam(A_{j,N})\le d_N$, we have $I_N\le \sup \{\Vert\psi_{\sigma_N}(\cdot-y)-\psi_{\sigma_N}\Vert_{L^q(\R^D)}:\Vert y\Vert_{\R^D}\le d_N \}\rightarrow0$ as $N\rightarrow\infty$. Applying Minkowski's inequality again, we have $II_N\le \int_{\R^D}\Vert p(\cdot-\sigma_N \mu)-p\Vert_{L^q(\R^D)}\psi(\mu)d\mu$. Then the continuity of translations $p\mapsto p(\cdot-y)$ with respect to $L^q$-topology and the dominant convergence theorem imply $\lim_{N\rightarrow\infty}II_N=0$.\qed

\noindent\textbf{Proof of Theorem \ref{thm: penalized MSD approximation}}: Since $p$ and $\psi$ are compactly supported, $\{\mu_{j,N}\}_{j=1}^N\subset \supp (p)$ for all $N$, and $\lim_{M\rightarrow\infty}\sigma_N=0$, there exists a compact $B$ containing the support of $p$, $\psi_{\sigma_N}(\cdot-\mu_{j,N})$, and $p_N(\cdot\vert\theta_N)=\sum_{j=1}^N\theta_{j,N}\psi_{\sigma_N}(\cdot-\mu_{j,N})$ for all $N$, and $f$ has no ambiguity point in $B$. Then $\left\vert \mathcal{K}_{\lambda,P_N}(f)-\mathcal{K}_{\lambda,p}(f)\right\vert\le H_N + I_N$, where
\begin{align*}
& I_N=\left\Vert p_N(\cdot\vert\theta_N)-p\right\Vert_{L^1(\R^D)}\times \sup_{x\in B}\left\Vert x-f\left(\pi_f(x) \right)\right\Vert_{\R^D}^2,\\
& H_N=\sum_{j=1}^N \left(\sup_{N': N'\ge j}\theta_{j,N'}\right)\times \left(H_{j,N}^*+H_{j,N}^{**}\right),\\
& H_{j,N}^*=\left\vert \int_{B} \left\Vert x-f\left(\pi_f(x) \right)\right\Vert_{\R^D}^2\left[\psi_{\sigma_N}(x-\mu_{j,N})dx-\delta_{\mu_{j}}(dx)\right]\right\vert \le 2\times \sup_{x\in B}\left\Vert x-f\left(\pi_f(x) \right)\right\Vert_{\R^D}^2,\\ 
& H_{j,N}^{**}=\left\vert \int_{B} \left\Vert x-f\left(\pi_f(x) \right)\right\Vert_{\R^D}^2\left[\delta_{\mu_{j,N}}(dx)-\delta_{\mu_{j}}(dx)\right]\right\vert \le 2\times \sup_{x\in B}\left\Vert x-f\left(\pi_f(x) \right)\right\Vert_{\R^D}^2,\ \ \ \mbox{for all $N$}.
\end{align*}
$p_N(\cdot\vert\theta_N)\rightarrow p$ in $L^1$ implies $\lim_{N\rightarrow\infty}I_N=0$. One can show $\lim_{N\rightarrow\infty}\mathcal{F}(\psi_{\sigma_N}(\cdot-\mu_{j,N}))=\lim_{N\rightarrow\infty}\mathcal{F}(\delta_{\mu_{j,N}})=\mathcal{F}(\delta_{\mu_{j}})$, where $\mathcal{F}$ denotes Fourier transform. Since the Fourier transform of a probability is the characteristic function of this probability, \textit{Levy continuity theorem} implies that the probability measure $\psi_{\sigma_N}(\cdot-\mu_{j,N})dx$ converges to $\delta_{\mu_{j}}$ weakly and $\delta_{\mu_{j,N}}$ converges to $\delta_{\mu_j}$ weakly as $N\rightarrow\infty$. Theorem \ref{measurability of projection indices} implies the continuity of $\Vert x-f(\pi_f(x) )\Vert_{\R^D}^2$ in $B$, and \textit{Portmanteau theorem} implies $H_{j,N}^*, H_{j,N}^{**}\rightarrow 0$ as $N\rightarrow\infty$ for all $j$. Then \textit{dominated convergence theorem} implies $\lim_{N\rightarrow\infty}H_N=0$. \qed

\begin{lemma}\label{Duchon, 1977, Theorem 4}
	(Theorem 4 bis of \cite{duchon1977splines}) Suppose $d\le3$. Let $\mathcal{C}$ be a finite subset of $\R^d$ such that every polynomial in $Poly_1[t]$ is uniquely determined by its values on $\mathcal{C}$. Then there exists exactly one function of the form $\sigma(t)=\sum_{c\in\mathcal{C}}s_c\eta_{4-d}(t-c)+p(t)$ taking prescribed values on $\mathcal{C}$, where $p\in Poly_1[t]$ and $\sum_{c\in\mathcal{C}}s_c q(c)=0$ for all $q\in Poly_1[t]$. Moreover, if $\gamma$ is another function taking the same prescribed values on $\mathcal{C}$, one has $\Vert\nabla^{\otimes 2}\sigma\Vert_{L^2}\le\Vert\nabla^{\otimes 2}\gamma\Vert_{L^2}$. 
\end{lemma}

\noindent\textbf{Proof of Theorem \ref{main functional theorem}}: Lemma \ref{Duchon, 1977, Theorem 4} implies that $g^*=\arg\min_{f\in\nabla^{-\otimes 2}L^2}\mathcal{K}_{\lambda, \widehat{Q}_N}(f, f_{(n)})$ is of the form (\ref{analytic expression of minimizer}). Theorem \ref{Sobolev embedding}, $d\le3$, and the form of (\ref{analytic expression of minimizer}) imply $g^*\in C_\infty\bigcap\nabla^{-\otimes 2}L^2$. Hence,
\begin{align*}
    g^*=\arg\min_{f\in C_\infty\bigcap\nabla^{-\otimes 2}L^2}\mathcal{K}_{\lambda, \widehat{Q}_N}\left(f, f_{{(n)}}\right)=f_{(n+1)}.
\end{align*}
\qed

\newpage

\vspace{-0.2in}
\bibliography{sample}

\begin{thebibliography}{29}
\providecommand{\natexlab}[1]{#1}
\providecommand{\url}[1]{\texttt{#1}}
\expandafter\ifx\csname urlstyle\endcsname\relax
  \providecommand{\doi}[1]{doi: #1}\else
  \providecommand{\doi}{doi: \begingroup \urlstyle{rm}\Url}\fi

\bibitem[Adams and Fournier(2003)]{adams2003sobolev}
R.~A. Adams and J.~J. Fournier.
\newblock \emph{Sobolev spaces. 2nd ed.}
\newblock Amsterdam: Academic Press, 2003.

\bibitem[Belkin and Niyogi(2003)]{belkin2003laplacian}
M.~Belkin and P.~Niyogi.
\newblock Laplacian eigenmaps for dimensionality reduction and data
  representation.
\newblock \emph{Neural Computation}, 15\penalty0 (6):\penalty0 1373--1396,
  2003.

\bibitem[Boissard et~al.(2015)Boissard, Le~Gouic, Loubes,
  et~al.]{boissard2015distribution}
E.~Boissard, T.~Le~Gouic, J.-M. Loubes, et~al.
\newblock Distribution’s template estimate with wasserstein metrics.
\newblock \emph{Bernoulli}, 21\penalty0 (2):\penalty0 740--759, 2015.

\bibitem[Dempster et~al.(1977)Dempster, Laird, and Rubin]{dempster1977maximum}
A.~P. Dempster, N.~M. Laird, and D.~B. Rubin.
\newblock Maximum likelihood from incomplete data via the em algorithm.
\newblock \emph{Journal of the Royal Statistical Society: Series B
  (Methodological)}, 39\penalty0 (1):\penalty0 1--38, 1977.

\bibitem[Do~Carmo(2016)]{do2016differential}
M.~P. Do~Carmo.
\newblock \emph{Differential geometry of curves and surfaces: revised and
  updated second edition}.
\newblock New York : Dover Publications, INC., 2016.

\bibitem[Duchamp et~al.(1996)Duchamp, Stuetzle, et~al.]{duchamp1996extremal}
T.~Duchamp, W.~Stuetzle, et~al.
\newblock Extremal properties of principal curves in the plane.
\newblock \emph{The Annals of Statistics}, 24\penalty0 (4):\penalty0
  1511--1520, 1996.

\bibitem[Duchon(1977)]{duchon1977splines}
J.~Duchon.
\newblock Splines minimizing rotation-invariant semi-norms in sobolev spaces.
\newblock In \emph{Constructive theory of functions of several variables},
  pages 85--100. Springer, 1977.

\bibitem[Dudek and Holly(1994)]{dudek1994nonlinear}
E.~Dudek and K.~Holly.
\newblock Nonlinear orthogonal projection.
\newblock In \emph{Annales Polonici Mathematici}, volume~59, pages 1--31.
  Instytut Matematyczny Polskiej Akademii Nauk, 1994.

\bibitem[Eloyan and Ghosh(2011)]{eloyan2011smooth}
A.~Eloyan and S.~K. Ghosh.
\newblock Smooth density estimation with moment constraints using mixture
  distributions.
\newblock \emph{Journal of Nonparametric Statistics}, 23\penalty0 (2):\penalty0
  513--531, 2011.

\bibitem[Enomoto and Okura(2013)]{enomoto2013total}
K.~Enomoto and M.~Okura.
\newblock The total squared curvature of curves and approximation by piecewise
  circular curves.
\newblock \emph{Results in Mathematics}, 64\penalty0 (1-2):\penalty0 215--228,
  2013.

\bibitem[Gerber and Whitaker(2013)]{gerber2013regularization}
S.~Gerber and R.~Whitaker.
\newblock Regularization-free principal curve estimation.
\newblock \emph{The Journal of Machine Learning Research}, 14\penalty0
  (1):\penalty0 1285--1302, 2013.

\bibitem[Hastie(1984)]{hastie1984principal}
T.~Hastie.
\newblock Principal curves and surfaces.
\newblock Technical report, Stanford University, California, Laboratory for
  Computational Statistics, 1984.

\bibitem[Hastie and Stuetzle(1989)]{hastie1989principal}
T.~Hastie and W.~Stuetzle.
\newblock Principal curves.
\newblock \emph{Journal of the American Statistical Association}, 84\penalty0
  (406):\penalty0 502--516, 1989.

\bibitem[Hauberg(2015)]{hauberg2015principal}
S.~Hauberg.
\newblock Principal curves on riemannian manifolds.
\newblock \emph{IEEE Transactions on Pattern Analysis and Machine
  Intelligence}, 38\penalty0 (9):\penalty0 1915--1921, 2015.

\bibitem[Jolliffe(2002)]{jolliffe1986principal}
I.~T. Jolliffe.
\newblock \emph{Principal component analysis. 2nd ed}.
\newblock New York: Springer, 2002.

\bibitem[K{\'e}gl et~al.(2000)K{\'e}gl, Krzyzak, Linder, and
  Zeger]{kegl2000learning}
B.~K{\'e}gl, A.~Krzyzak, T.~Linder, and K.~Zeger.
\newblock Learning and design of principal curves.
\newblock \emph{IEEE Transactions on Pattern Analysis and Machine
  Intelligence}, 22\penalty0 (3):\penalty0 281--297, 2000.

\bibitem[Kim et~al.(2020)Kim, Lee, and Oh]{kim2020spherical}
J.-H. Kim, J.~Lee, and H.-S. Oh.
\newblock Spherical principal curves.
\newblock \emph{arXiv preprint arXiv:2003.02578}, 2020.

\bibitem[Kirov and Slep{\v{c}}ev(2017)]{kirov2017multiple}
S.~Kirov and D.~Slep{\v{c}}ev.
\newblock Multiple penalized principal curves: Analysis and computation.
\newblock \emph{Journal of Mathematical Imaging and Vision}, 59\penalty0
  (2):\penalty0 234--256, 2017.

\bibitem[Koenker and Mizera(2004)]{koenker2004penalized}
R.~Koenker and I.~Mizera.
\newblock Penalized triograms: Total variation regularization for bivariate
  smoothing.
\newblock \emph{Journal of the Royal Statistical Society: Series B (Statistical
  Methodology)}, 66\penalty0 (1):\penalty0 145--163, 2004.

\bibitem[Lindsay(1983)]{lindsay1983geometry}
B.~G. Lindsay.
\newblock The geometry of mixture likelihoods: A general theory.
\newblock \emph{The Annals of Statistics}, pages 86--94, 1983.

\bibitem[{R Core Team}(2019)]{citationR}
{R Core Team}.
\newblock \emph{R: A Language and Environment for Statistical Computing}.
\newblock R Foundation for Statistical Computing, Vienna, Austria, 2019.
\newblock URL \url{https://www.R-project.org}.

\bibitem[Roweis and Saul(2000)]{roweis2000nonlinear}
S.~T. Roweis and L.~K. Saul.
\newblock Nonlinear dimensionality reduction by locally linear embedding.
\newblock \emph{Science}, 290\penalty0 (5500):\penalty0 2323--2326, 2000.

\bibitem[Rudin(1991)]{rudin1991functional}
W.~Rudin.
\newblock \emph{Functional analysis. 2nd ed.}
\newblock New York: McGraw-Hill, 1991.

\bibitem[Seguy and Cuturi(2015)]{seguy2015principal}
V.~Seguy and M.~Cuturi.
\newblock Principal geodesic analysis for probability measures under the
  optimal transport metric.
\newblock \emph{Advances in Neural Information Processing Systems},
  28:\penalty0 3312--3320, 2015.

\bibitem[Smola et~al.(2001)Smola, Mika, Sch{\"o}lkopf, and
  Williamson]{smola2001regularized}
A.~J. Smola, S.~Mika, B.~Sch{\"o}lkopf, and R.~C. Williamson.
\newblock Regularized principal manifolds.
\newblock \emph{Journal of Machine Learning Research}, 1\penalty0
  (Jun):\penalty0 179--209, 2001.

\bibitem[Tenenbaum et~al.(2000)Tenenbaum, De~Silva, and
  Langford]{tenenbaum2000global}
J.~B. Tenenbaum, V.~De~Silva, and J.~C. Langford.
\newblock A global geometric framework for nonlinear dimensionality reduction.
\newblock \emph{Science}, 290\penalty0 (5500):\penalty0 2319--2323, 2000.

\bibitem[Tibshirani(1992)]{tibshirani1992principal}
R.~Tibshirani.
\newblock Principal curves revisited.
\newblock \emph{Statistics and Computing}, 2\penalty0 (4):\penalty0 183--190,
  1992.

\bibitem[Wahba(1990)]{wahba1990spline}
G.~Wahba.
\newblock \emph{Spline models for observational data}, volume~59.
\newblock SIAM, 1990.

\bibitem[Yue et~al.(2016)Yue, Zipunnikov, Bazin, Pham, Reich, Crainiceanu, and
  Caffo]{yue2016parameterization}
C.~Yue, V.~Zipunnikov, P.-L. Bazin, D.~Pham, D.~Reich, C.~Crainiceanu, and
  B.~Caffo.
\newblock Parameterization of white matter manifold-like structures using
  principal surfaces.
\newblock \emph{Journal of the American Statistical Association}, 111\penalty0
  (515):\penalty0 1050--1060, 2016.

\end{thebibliography}

\end{document}